\documentclass{article}
\usepackage[margin=1in]{geometry}

\usepackage{amsmath,amssymb,amsfonts,amsthm}

\usepackage{natbib}
 \bibpunct[, ]{(}{)}{,}{a}{}{,}
 \def\bibfont{\small}
 
 \def\bibhang{24pt}

\usepackage{tikz}
\usetikzlibrary{shapes.geometric, arrows}
\usepackage{rotating}
\usepackage{fancyvrb}
\usepackage{multirow}

\providecommand{\keywords}[1]
{
  \small	
  \textbf{\textit{Keywords---}} #1
}
\newtheorem{definition}{Definition}
\newtheorem{theorem}{Theorem}
\newtheorem{lemma}[theorem]{Lemma}
\newtheorem{example}{Example}
\newtheorem{corollary}{Corollary}

\date{}

\begin{document}

\title{Statistically Optimal Uncertainty Quantification for Expensive Black-Box Models}

\author{Shengyi He$^{1}$, Henry Lam$^{1}$  \\
        \small $^{1}$ Department of Industrial Engineering and Operations Research, Columbia University \\
}

\maketitle
\begin{abstract}
{
 Uncertainty quantification, by means of confidence interval (CI) construction, has been a fundamental problem in statistics and also important in risk-aware decision-making. In this paper, we revisit the basic problem of CI construction, but in the setting of \emph{expensive black-box} models. This means we are confined to using a low number of model runs, and without the ability to obtain auxiliary model information such as gradients. In this case, there exist classical methods based on data splitting, and newer methods based on suitable resampling. However, while all these resulting CIs have similarly accurate coverage in large sample, their efficiencies in terms of interval length differ, and a systematic understanding of which method and configuration attains the shortest interval appears open. Motivated by this, we create a theoretical framework to study the statistical optimality on CI tightness under computation constraint. Our theory shows that standard batching, but also carefully constructed new formulas using uneven-size or overlapping batches, batched jackknife, and the so-called cheap bootstrap and its weighted generalizations, are statistically optimal. Our developments build on a new bridge of the classical notion of uniformly most accurate unbiasedness with batching and resampling, by viewing model runs as asymptotically Gaussian ``data", as well as a suitable notion of homogeneity for CIs.}
 \end{abstract}

\maketitle

\section{Introduction}
\label{sec:intro}

Uncertainty quantification, by means of confidence interval (CI) construction, has been a fundamental problem in statistics and also important in risk-aware decision-making. In this paper, we revisit the basic problem of CI construction, but in the setting of \emph{expensive black-box} models. This means that we are confined to using a possibly extremely low number of model runs, and we also lack the immediate capability to obtain auxiliary model information such as gradients. The option, then, is to obtain a CI using very few model runs on the data set or some modifications on it. In this paper, our main goal is to understand whether, or more specifically what is the best way, to achieve this.

The above question is motivated from the growing scale in modern data science, but also inference on large computational simulation models. For example, in high-dimensional regression problems, we might be interested in interval estimation of a certain coefficient or prediction value at a test point. Each fitting might involve running a large-scale optimization routine that can be expensive. Another example, which is important in simulation modeling, is the so-called input uncertainty problem (\citealt{henderson2003input,chick2006subjective,Barton2012tutorial,song2014advanced,Lam2016advanced,CORLU2020100162,Barton2022}), which addresses the interplay between statistical calibration error and Monte Carlo computation error. To explain, a stochastic simulation model typically involves generating random variates from input distributions that are fed into the system logic, and producing outputs for downstream decision-making. When the input parameters or distributions are not fully known but instead calibrated from external data, these calibration noises can propagate to affect the output accuracy, constituting what is known as input uncertainty. In this case, the CIs for target output quantities would need to account for both the input statistical noise and the Monte Carlo noise in the stochastic simulator. For example, suppose we are interested in steady-state estimation under input uncertainty. Then each ``model run" would entail simulating the system long enough to wash away the initialization bias (\citealt{glynn1992experiments,alexopoulos1998output}). Alternatively, suppose we are interested in a transient output measure, e.g., the expected waiting time for a certain time horizon in a large queueing system. Then each ``model run" would involve averaging many simulated trajectories to wash away the stochastic or aleatory noises. That is, if the model evaluation is conducted by Monte Carlo computation, and the statistical uncertainty comes from the input data, then a valid inference would require each model evaluation to bear sufficient computation in order to eliminate the model bias or aleatory variability . For sophisticated stochastic simulators, each such evaluation can be expensive and by-and-large black-box.

To understand the challenges, we first point out that in the above setting, common inference approaches could encounter difficulties one way or another. The delta method or infinitesimal jackknife (\citealt{efron1981nonparametric}), which utilizes a normality interval with a plug-in estimate of the standard error, relies on gradient information with respect to the parameters or input distributions which is not obtainable in the black-box setting. Resampling approaches, such as the bootstrap and jackknife (\citealt{efron1981nonparametric,efron1994introduction,Davison_Hinkley_1997,Shao1995}), typically requires implementations that comprise repeated model runs for a sufficient number of times (the number of resamples in the case of bootstrap, and the number of leave-one-out estimates or equivalently the data size in the case of jackknife). On the other hand, there do exist lower-computation alternatives that serve exactly our considered goal, namely construct valid-coverage CIs with only very few model runs. These include data splitting approaches (\citealt{Schmeiser1982batch,NAKAYAMA20071330,glynn1990simulation,Schruben1983confidence}), which group data into several (possibly overlapping) batches and construct intervals based on the resulting batch estimates. These batching methods are historically motivated from the handling of serial dependent data arising in stochastic simulation (\citealt{fishman1978grouping,law1979sequential}) or Markov chain Monte Carlo (\citealt{geyer1992practical,jones2006fixed,flegal2010batch}), but they are also computationally light when using a small number of batches. More recently, the so-called cheap bootstrap (\citealt{ll2023,lam2022cheap,lam2022cheap1}) advocates CI construction by using a small number of resamples that are randomly drawn instead of deterministic grouping of the data. On a high level, these methods bypass the accurate estimation of the variability which is central, and also costly, to the delta method and classical resampling. Instead, they hinge on the construction of pivotal $t$-statistics that cancel out the variability. Despite the attractive computational saving, the price to pay in exchange compared to the traditional methods is the longer intervals, which stems from the inflated uncertainty in the implicit variance estimator and manifests in the use of $t$ instead of normal critical values in the intervals.

With the above background, we ask in this paper the following question: \emph{Given a fixed computation budget in terms of the number of model runs, what is a statistically optimal CI?} While all the low-computation CI construction methods produce similarly accurate coverage in large sample, their interval lengths (i.e., the statistical price mentioned above) differ, and a systematic understanding of which method and procedural configuration is the most efficient in terms of shortest interval length  appears open to our best knowledge. Here, when speaking of \emph{statistical optimality}, we mean the minimization of the asymptotic expected interval length, where the asymptotic is as data size grows. This criterion is intuitive, arguably more so than other possibilities in the literature. For example, the variance or the mean-squared error of the variance estimator \citep{song1995optimal} can be translated into the variance of CI length and is advantageously estimable using data, but does not result in a theoretical understanding of the optimal form of CIs which we aim for. Next, by a fixed computation budget, we mean the number of times we are allowed to run the model. More specifically, we consider the general setting where the target quantity is a functional of an input distribution. Given the input distribution, evaluating the target quantity is resource-consuming, but is not significantly affected by the choice of the input distribution. For instance, if the input distribution is an empirical distribution, the data size does not significantly matter. This happens in, e.g., steady-state estimation under input uncertainty, where each model run comprises a steady-state simulation driven by input distributions created from the data. Here, the runtime of the steady-state simulation could be insensitive to these input distributions and depend primarily on the system dynamic.

Our main contribution is to create a theoretical framework to answer the above question on CI statistical optimality under a small computation budget. Our theory examines the optimality of existing CIs as well as leads to \emph{new} CIs that are guaranteed to be shorter than their counterparts in the literature. More precisely, low-computation methods generically consists of two stages. In Stage 1, we determine a collection of input distributions  at which the output functional is evaluated, where the number of input distributions in this collection cannot exceed the computation budget. These input distributions are created from the data with a mechanism, such as subsamples or batches deterministically selected in batching methods, or randomly selected or weighted in the cheap bootstrap and its generalization. Then, in Stage 2, an interval formula is used to combine these evaluations. Under this setup, our results entail both \emph{local} and \emph{global} optimality. At the local level, we characterize, for a given Stage 1 mechanism in selecting the input distributions, the optimal CI formula to combine their evaluations. In particular, our results imply that, when data are divided into equal-size non-overlapping batches, the standard batching CI available in the literature is optimal among any other formulas to combine these batches. When data are resampled uniformly with replacement, the cheap bootstrap CI is again optimal among other formulas to combine the resample estimates. However, in the general case where the batches are unequal-sized or overlapped, our results give rise to new CI formulas that can differ from existing proposals \citep{su2023overlapping} and are guaranteed to be shorter asymptotically. At the global level, we show that the shortest CIs, regardless of the Stage 1 mechanism, are given by asymptotically equivalent $t$-based intervals, which include standard batching, cheap bootstrap, batched jackknife, as well as our newly derived unequal-size or overlapping batching. Lastly, we note that our results apply not only to i.i.d. data, but also to dependent data under suitable conditions that guarantee a proper central limit theorem (CLT).

Our roadmap to obtain the above optimality characterizations center around a new bridge of low-computation CIs and their interval lengths with the  classical notion of \emph{uniformly most accurate (UMA) unbiasedness} \citep{lehmann2005testing}. UMA unbiasedness for a CI is equivalent to having the CI's dual hypothesis test being uniformly most powerful (UMP) unbiased. Despite its long establishment, its connection with the efficiency of batching and resampling under low-computation environment has not been exploited to our knowledge, and doing so requires several novel developments. First is the view of Stage 1 model evaluations as asymptotically Gaussian, possibly heterogeneous and dependent ``data" in the UMA framework. Depending on our original data and our mechanism in selecting the input distributions, the model evaluations from Stage 1 exhibit a corresponding joint CLT as data size grows. These evaluations are then regarded as approximately Gaussian ``data" to construct a CI and the dual hypothesis test. Importantly, this test involves only two parameters, the unknown Gaussian mean and the unknown scaling of the covariance matrix that is otherwise determined from our Stage 1 scheme. This subsequently allows us to leverage Neyman-Pearson-type statistical tools, including completeness and Basu's theorem, to obtain approximately UMA unbiased CIs and the tightness equivalence of $t$-based intervals. Second, using a suitable combination of integration argument and the property of unbiasedness, we show that asymptotic UMA implies the shortest asymptotic expected length among unbiased CIs. This provides a direct translation from UMA to CI interval length that, to our best knowledge, has not been explicitly derived and utilized in the literature. Third, to precisely connect the asymptotic Gaussian approximation derived from our Stage 1 CLT with the exact inference tools in UMA, we introduce the notion of \emph{homogeneous} CIs and argue our asymptotic optimality within this CI class.  Homogeneity here means the CI upper and lower limits translate and scale in the same amount as those of the model evaluations. Intuitively, this implies the CIs cannot use knowledge on the mean and variability of the evaluations a priori, which coincides with our black-box presumption. Technically, this implies that even if we only assume a CLT at the ground-truth distribution, with a translation and rescaling argument, we are able to retrieve coverage probabilities of CIs under alternative hypotheses whose limiting distributions have different means and variances. These alternative hypotheses in turn play an important role in arguing UMP unbiasedness and subsequently the attainment of shortest intervals.

The rest of this paper is as follows. Section \ref{sec: lit} reviews related literature. Section \ref{sec: setup} introduces our framework for low-computation uncertainty quantification and defines the notion of optimality. Section \ref{sec: AUMAU} derives both locally optimal CIs under different Stage 1 schemes and globally optimal CIs. Section \ref{sec: numerics} validates our theory via numeric experiments and compares the performances among different low-computation CIs. Section \ref{sec: future work} concludes the paper and discusses future directions. The Appendix contains the proofs of all statements and additional technical lemmas, and documents some existing useful results.

\section{Literature Review}\label{sec: lit}
We divide our literature review into two parts, existing techniques that are or can be converted into low-computation CI construction methods (Section \ref{sec:techniques}), and the theoretical tools in our analyses (Section \ref{sec:tools}). 
\subsection{Existing Low-Computation CI Construction Methods}\label{sec:techniques}
Existing methods to construct CIs with low computation budget include data splitting or batching methods, and the cheap bootstrap. We will overview the related literature on both, and then briefly discuss other related methods that save computation from other perspectives than this paper. 

Batching methods were historically proposed to handle serially dependent data in simulation output analysis \citep{Schmeiser1982batch,munoz1997batch,fishman1978grouping,law1979sequential} and Markov chain Monte Carlo \citep{geyer1992practical,jones2006fixed,flegal2010batch}, especially for steady-state estimation. The idea is to group the data into (possibly overlapping) batches, and suitably aggregate these batch estimates via $t$-pivotal statistics to cancel out the variance. In the steady-state estimation literature, a stream of works studies automated procedures to determine hyperparameters such as the number of batches so as to satisfy user-specified CI coverage and accuracy, and reduce finite sample errors caused by the initialization, the skewness or the correlation of the data \citep{alexopoulos1996implementing,asap3,Tafazzoli2011nskart}. Extensions such as CIs for steady-state quantiles have also been studied \citep{Alexopoulos2014sequest}. To take care of the half-width, these works propose to increase the sample size and the number of batches until the CI half-width is shorter than a user-specified precision requirement. That is, they do not study how to get a shorter CI without increasing the computation budget as we do. 

The works above focus on non-overlapping batches. As an alternative, overlapping batching is sometimes preferred due to its flexibility in choosing the overlapping schemes. It is shown in \cite{meketon1984overlapping} that for steady-state mean estimation, the variance estimator using overlapping batching has a smaller variance than non-overlapping batching. The idea of overlapping can also be further generalized to the more general umbrella method of standardized time series (STS) \citep{glynn1990simulation,Schruben1983confidence,Goldsman1990new,calvin2006permuted}. Like batching, STS cancels out the variance term by taking the ratio between a point estimator and a variance estimator, but in a more general way where the variance estimator is a finite-sample approximation of a general functional of a Brownian motion that is asymptotically independent of the point estimator. Nonetheless, the overlapping batching in the above works only considers the case where 
 the starting points of adjacent batches differ by one, so the number of batches is large which makes these methods computationally expensive. Indeed, most of these works focus on the mean estimation problem, for which a trick is to compute the next batch mean by updating from the previous batch mean (page 229 of \cite{meketon1984overlapping}). However, for general models, such a trick is not available, and the overlapping batching and STS proposed in these papers would face heavy computation. Therefore, these methods are not considered as low-computation methods in this paper. More recently, \cite{su2023overlapping} consider more general overlapping batching where the difference between the starting indexes of adjacent batches can differ by a proportion of the batch size, and investigate the corresponding asymptotic distributions in the estimation of general functionals. The approach in \cite{su2023overlapping} can be low-computation, to which we will compare and, in particular, show both theoretically and empirically that ours is able to produce shorter intervals in general.
 
Despite the historical focus in serially dependent problems, batching methods can be used for low-computation CI construction by using a very small number of batches \citep{he2021higher,glynn2018constructing}, which we utilize in this work. Another type of low-computation methods is the recently proposed cheap bootstrap \citep{ll2023,lam2022cheap}. This approach, in some sense, can be viewed as a replacement of deterministic data splitting in batching methods by random resampling. Unlike classical bootstraps, the cheap bootstrap utilizes the sample-resample independence implied from the bootstrap CLT, instead of a direct approximation of the resample distribution to the original sampling distribution, which in turn allows to construct $t$-pivotal statistics as in batching methods. 
\cite{lam2022cheap1} and \cite{lam2023resampling} further apply the cheap bootstrap to CI construction for simulation input uncertainty and stochastic gradient descent respectively.

Finally, we briefly note that subsampling methods can serve to save computation cost, but from another perspective that is not considered in this paper. The basic idea of these methods is to reduce the number of distinct data points in the resampled data, which can reduce computation cost per model run in problems such as M-estimation. These methods include the $m$-out-of-$n$ bootstrap \citep{politis1999subsampling,bickel1997resampling}, bag-of-little bootstrap \citep{kleiner2014scalable} and subsampled double bootstrap \citep{sengupta2016subsampled}.

\subsection{Theoretical Tools}\label{sec:tools}

We utilize the notion of UMA unbiasedness in our analyses. First, recall that in hypothesis testing, the power of a test refers to the probability of rejection when the data comes from an alternative distribution not covered by the null hypothesis. A test is unbiased if for every alternative distribution, the power is greater than the size (the upper limit for the rejection probability when null hypothesis holds). A test is UMP unbiased if it maximizes the power at every alternative distribution among unbiased tests. A similar line of definitions hold for CIs using their dual relation with hypothesis tests: A CI is UMA unbiased if it minimizes the coverage probability at every alternative points among unbiased CIs, where unbiasednss means that the coverage at every alternative points needs to be smaller than the nominal coverage level. The dual relation between UMP unbiased hypothesis tests and UMA unbiased CIs is discussed in, e.g., Section 5.5 of \cite{lehmann2005testing}, while Section 4.4 of \cite{lehmann2005testing} provides techniques to find UMP unbiased tests for exponential families. In general, finding a UMP unbiased test can be viewed as an infinite-dimensional optimization problem. With the structure of exponential family, the problem can be reduced by conditioning on nuisance parameters, which is much easier to solve. The uniformity among a class of alternative distributions also presents similarity with local minimax theorems (e.g., Section 3.11 of \cite{WeakConvergence_VdV}), where the optimality is stated for the worst-case performance when the alternative distribution varies in a local set.

Next, we contrast our analyses with the theoretical criteria and comparison results of batching methods and the cheap bootstrap in existing works. For the CI half width, \cite{Schmeiser1982batch} shows that if we assume the data size is large enough so that the non-normality of the batch estimators is negligible, then the expected half width would decrease as the number of batches increases, but the rate of decrease would become much slower when the number of batches is large. Similar observations are also made in \cite{glynn2018constructing,lam2022cheap}.  In contrast to these works, we analyze with a fixed rather than allowing a varying computation budget. \cite{glynn1990simulation} also studies the length of CI. However, they consider the shortest CI in a restricted family where the CIs have the same pivotal statistics and only differ by the level of asymmetry when choosing the two quantiles of the pivotal statistics as threshold values. Another performance criterion that has been studied is the coverage error. \cite{pope1995improving} analyzes the coverage error of sectioning (a variant of batching) using Edgeworth expansions when the number of batches goes to infinity. \cite{he2021higher} studies coverage errors for several batching variants with a fixed number of batches, by deriving Edgeworth expansions for $t$-statistics and the computation procedures for the involved coefficients. \cite{lam2022cheap} studies  coverage errors for the cheap bootstrap, and \cite{ll2023} studies its finite-sample behaviors and shows that the coverage error remains controlled even as the problem size grows closely with the data size. 
In addition, \cite{lewis1989simulation} uses the jackknife to reduce small-sample bias within sections, but at the cost of greater computation time and uncertainty regarding variance inflation. \cite{song1995optimal,flegal2010batch} derive batch sizes that minimize the mean squared errors of batching variance estimators for steady-state estimation problems. Compared to these works, we use a different notion of optimality in terms of the power of hypothesis tests and the expected length, which is arguably more interpretable and provides insights that cannot be explained by the other criteria reviewed above.

\section{Theoretical Framework on Low-Computation Uncertainty Quantification Methods}\label{sec: setup}
Suppose that we are interested in estimating a real-valued functional $\psi(P)$ given samples $X_1,\dots,X_{n}$ drawn from an unknown distribution $P$. The goal is to construct a CI for $\psi(P)$ at a prescribed level $1-\alpha$. More precisely, we aim for a CI, say $[L,U]$, that is \emph{asymptotically valid}, i.e.,
$$\liminf_{n\to\infty}P(\psi(P)\in[L,U])\geq 1-\alpha,$$
 where the outer probability $P$ is taken with respect to the data in constructing $[L,U]$. Relatedly, we say that $[L,U]$ is \emph{asymptotically exact} if $$\lim_{n\to\infty}P(\psi(P)\in[L,U])= 1-\alpha.$$ It can be checked that asymptotically exact CIs are always asymptotically valid. Conversely, for a CI to be shortest among asymptotically valid CIs, it typically needs to be asymptotically exact.

We are interested in the case where $\psi(\cdot)$ is a black-box functional that is computationally expensive to compute. More concretely, we are confined to a simulation budget $K$, which means that we can only evaluate $\psi(\cdot)$ for at most $K$ times, and are not capable of generating auxiliary information such as gradient information to assist with CI construction. For example, for the input uncertainty quantification problem discussed in the introduction, $\psi(P)$ could be the steady-state mean for the queue length in a queueing system where the interarrival times has distribution $P$. The data could be the past interarrival times observed from the unknown $P$. The model is computationally expensive in the sense that for any input distribution $Q$, we need to simulate a long process to wash out the initialization bias when estimating $\psi(Q)$. As another example, $\psi(P)$ could also be a transient measure such as the expected waiting time of a certain customer. In this case, the heavy computation comes from the need to simulate the queue for a sufficient number of times to eliminate the stochastic variability. In these situations, there is a high computational cost in evaluating $\psi(Q)$, but this cost is insensitive to the choice of $Q$. This constitutes precisely our setup, where we measure the computation cost as the number of evaluations for $\psi(\cdot)$, and regard the cost as insensitive to the distribution $Q$ that we feed into $\psi(\cdot)$. 

\subsection{A Two-Stage Dissection}
The above setup essentially means that the CI can only be constructed from at most $K$ evaluations of $\psi(\cdot)$ at different input distributions created from the data, \emph{but nothing more}. Such procedures in general can be decomposed into two stages: In Stage 1, we obtain estimates $\mathbf{Y}_{n}=(\psi_1,\dots,\psi_K)$ where each of $\psi_i,i=1,2,\dots,K$ is constructed (differently) from $X_1,\dots,X_{n}$. Then, in Stage 2, we combine the Stage 1 estimates to get a CI  $[L(\mathbf{Y}_{n}),U(\mathbf{Y}_{n})]$. Figure \ref{fig: two stages} illustrates this two-stage process.

\tikzstyle{startstop} = [rectangle, rounded corners, 
minimum width=3cm, 
minimum height=1cm,
text centered, 
draw=black]

\tikzstyle{arrow} = [thick,->,>=stealth]

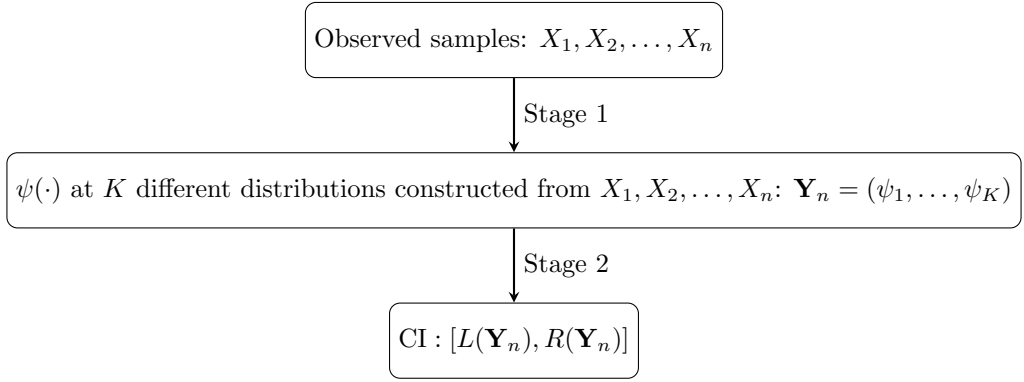
\begin{figure}
    \centering
    \begin{tikzpicture}[node distance=2cm]

\node (start) [startstop] {Observed samples: $ X_1, X_2,\dots,X_n$};
\node (in1) [startstop, below of=start] {$\psi(\cdot)$ at $K$ different distributions constructed from $ X_1, X_2,\dots,X_n$: $\mathbf{Y}_n=(\psi_1,\dots,\psi_K) $};
\node (pro1) [startstop, below of=in1] {$\text{CI}: [L(\mathbf{Y}_n),R(\mathbf{Y}_n)] $};

\draw [arrow] (start) -- node[anchor=west] {Stage 1} (in1);
\draw [arrow] (in1) -- node[anchor=west] {Stage 2}(pro1);

\end{tikzpicture}
    \caption{Dissection of general low-computation CI construction methods.}
    \label{fig: two stages}
\end{figure}

To understand the above further, here are several exemplifying approaches we can utilize from the literature:

\begin{itemize}
    \item \emph{Standard batching:} \sloppy Here, $\psi_i=\psi(\hat P_n^i)$ where $\hat P_n^i$ is the empirical distribution constructed from the $i$-th batch of data, $X_{(i-1)n/K+1},\dots,X_{in/K}$. Given $\mathbf{Y}_{n,B}=(\psi_1,\dots,\psi_K)$, the standard batching CI is
$$CI_{B}(\mathbf{Y}_{n,B}) := \bar\psi \pm t_{K-1,1-\alpha/2}\frac{S_{B}}{\sqrt{K}}$$
where $\bar\psi = \frac{\sum_{i=1}^K\psi_i}{K}$ and $S_B^2=\frac{1}{K-1}\sum_{i=1}^{K}\left(\psi_i-\bar\psi\right)^{2}$ are the sample mean and variance of the batch estimates, and $t_{K-1,1-\alpha/2}$ is the upper $\alpha/2$-quantile of $t_{K-1}$, the $t$-distribution with degree of freedom $K-1$. 

\item \emph{Overlapping batching:} Using the version in \cite{su2023overlapping}, we first assign one batch to be the entire data set, and each other batch has equal sizes denoted by $\gamma n$ and adjacent batches have equal overlaps or distances. More concretely, the first batch is the entire sample $X_1,\ldots,X_n$. For $2\leq j\leq K$, batch $j$ starts from the $(j-2)\frac{n-\gamma n}{K-2}+1$-th and ends at the $(j-2)\frac{n-\gamma n}{K-2}+\gamma n$-th data point. As can be checked, a batch could overlap with other batches when $\gamma>1/(K-1)$. Let $\mathbf{Y}_{n,OB}=(\psi_{j,OB},1\leq j\leq K)$ denote the batch estimates. Then, we use the overlapping batching CI:
\[
CI_{OB}(\mathbf{Y}_{n,OB}):=\psi_{1,OB}\pm c_{\gamma,K-1,1-\alpha/2}S_{OB}.
\]
Here, $S_{OB}^2={\frac{\gamma}{1-\gamma}\frac{1}{K-1}\sum_{j=2}^{K}\left(\psi_{j,OB}-\psi_{1,OB}\right)^{2}}$ and $c_{\gamma,K-1,1-\alpha/2}$ is the upper $\alpha/2$-quantile of the OB-I limit proposed in \cite{su2023overlapping}.

\item \emph{Cheap bootstrap:} We resample, i.e., sample with replacement, from $X_1,\ldots,X_n$ to obtain the resampled data set $X_1^{*b},\ldots,X_n^{*b}$, independently for $K$ times. Then $\mathbf{Y}_{n,CB}=(\psi(\hat{P}_n),\psi(P^{*1}_n),\dots,\psi(P^{*(K-1)}_n))$ where $\hat{P}_n$ is the empirical distribution of $\{X_1,\dots,X_n\}$ and  $P^{*b}_n,b=1,2,\dots,K-1$ is the empirical distribution of $\{X_1^{*b},\dots,X_n^{*b}\}$. The cheap bootstrap CI is constructed as 
\[
CI_{CB}\left(\mathbf{Y}_{n,CB}\right):=\psi(\hat{P}_n)\pm t_{K-1,1-\alpha/2}S_{CB}
\] 
where $S^2_{CB} = \frac{1}{K-1}\sum_{b=1}^{K-1} (\psi(P^{*b}_n)-\psi(\hat{P}_n))^2$.

\item \emph{Weighted cheap bootstrap:} Cheap bootstrap can also be generalized as follows. Let $W_n=(W_{1,n},W_{2,n},\dots, W_{n,n})$ be a nonnegative random weight vector whose elements are exchangeable and satisfy $W_{1,n}+\dots+W_{n,n}=1$. Denote the distribution of $W_n$ as $\mathcal{W}_n$. Then, we can construct $\mathbf{Y}_{n,WCB}=(\psi(\hat{P}_n),\psi(P^{W,1}_n),\dots,\psi(P^{W,K-1}_n))$ where $P^{W,b}_n = \sum_{i=1}^n W_{i,n}^{(b)}\delta_{X_i}$, $\delta_{X_i}$ stands for point mass at $X_i$, and each of $(W_{1,n}^{(b)},W_{2,n}^{(b)},\dots, W_{n,n}^{(b)}),b=1,2,\dots,K-1$ is drawn independently from $\mathcal{W}_n$. 
The weighted cheap bootstrap CI is constructed as 
\[
CI_{WCB}\left(\mathbf{Y}_{n,WCB}\right) := \psi(\hat{P}_n) \pm t_{K-1,1-\alpha/2}\sigma_W^{-1}S_{WCB}.
\]
Here, $\sigma_W^2 = \lim_{n\to\infty}Var(nW_{1,n})$ and $S_{WCB}^2 = \frac{1}{K-1}\sum_{b=1}^{K-1} (\psi(P^{W,b}_n)-\psi(\hat{P}_n))^2$. For example, for the original cheap bootstrap, $n\mathcal{W}_n$ is multinomial with $n$ trials and $n$ categories. Therefore, $Var(nW_{1,n})= (n-1)/n$ and $\sigma_W^2 = \lim_{n\to\infty} (n-1)/n = 1$. Comparing with the formula for $CI_{CB}$, we see that the cheap bootstrap is a special case of the weighted cheap bootstrap where $\mathcal{W}_n$  is constructed from multinomial distribution. 
\end{itemize}

We point out that, under regularity conditions regarding the joint asympototic normality of Stage 1 estimates, all the above methods produce  asymptotically exact CIs through a pivotal statistic construction. For example, in the case of standard batching, we consider the pivotal statistic $\frac{\sqrt{K}\left(\bar{\psi}-\psi(P)\right)}{S_{B}}$ and the event $\psi(P)\in CI_B$ is equivalent to \[-t_{K-1,1-\alpha/2}\leq\frac{\sqrt{K}\left(\bar{\psi}-\psi(P)\right)}{S_{B}}\leq t_{K-1,1-\alpha/2}\]
Under joint asymptotic normality of $\mathbf{Y}_{n,B}$, the pivotal statistic $\frac{\sqrt{K}\left(\bar{\psi}-\psi(P)\right)}{S_{B}}$ converges to a $t_{K-1}$-distribution. Therefore, the limiting coverage probability is $1-\alpha$ as claimed 
Similar arguments hold for the other methods. More concretely, for $CI_{OB}$, its asymptotic exactness is shown in \cite{su2023overlapping} under a set of conditions they proposed. For the other methods, we will provide explicit conditions for their asymptotic exactness as well as other properties later in this paper. Although all of these methods achieve asymptotically exact coverage, their asymptotic lengths can differ, and our goal is to find the CIs with the shortest asymptotic expected length.

\subsection{CI Homogeneity}
To build our framework in obtaining the shortest CIs, we need to properly define the class of CIs that we compare against. Note that one of the motivations of batching and cheap bootstrap methods is that the variability of the estimator is unknown, so it needs to be cancelled out via constructing pivotal statistics. Correspondingly, the broader CI class that we consider should not use knowledge on variability. To this end, we make the following definition. 
\begin{definition}\label{def: homogenous}
\textit{(Homogeneous CI)} 
    We say that $[L(\cdot),U(\cdot)]$ is \textit{homogeneous} if $L(\cdot)$ and $U(\cdot)$ satisfy
\begin{enumerate}
\item $L(c\textbf{x})=cL(\textbf{x}),U(c\textbf{x})=cU(\textbf{x})$ for any constant $c>0$.
\item $L(\textbf{x}+c1_{K})=L(\textbf{x})+c,U(\textbf{x}+c1_{K})=U(\textbf{x})+c$ for any constant $c\in\mathbb{R}$.
\item 
Both $L(\cdot)$ and $U(\cdot)$ are continuous.
\end{enumerate}

\end{definition}
 Denote the class of homogeneous $[L(\cdot),U(\cdot)]$ as $\mathcal{H}$. In Definition \ref{def: homogenous}, condition 3 is a mild smoothness condition. Conditions 1-2 are also natural for a reasonable CI: When each Stage 1 
estimate is scaled or shifted by the same amount, the CI is also
scaled or shifted by that amount. Importantly, conditions 1-2 exclude CIs that use information on the true mean or variance.  For example, for standard batching, $[L(\textbf{x}),U(\textbf{x})] = \bar{{x}}\pm t_{K-1;1-\alpha/2}\sqrt{\sum_{i=1}^K\left(x_i-\bar{x}\right)^2/(K-1)}/\sqrt{K}$ satisfies these conditions. But on the other hand, the normality CI $[L(\textbf{x}),U(\textbf{x})] = \bar{{x}}\pm z_{1-\alpha/2}\sigma_0/\sqrt{n}$, where we assume we know the asymptotic variance $\sigma_0$ such that $\sqrt{n}(\psi(\hat{P}_n)-\psi(P))\Rightarrow N(0,\sigma_0^2)$, does not satisfy condition 1. If we only have condition 1 but not 2, then $L(x)\equiv U(x)\equiv 0$ would belong to this family. When $\psi(P)$ happens to be 0, this singleton CI has 100\% coverage and 0 length, so it is optimal. However, this is not the CI we are looking for, and condition 2 serves to exclude such CIs that use the true mean information.

Technically, with these conditions, once
we have the behavior of the CI under one distribution, we can get
its behavior when the distribution is perturbed through a mean shift or rescaling, which will be useful in our leveraging of hypothesis testing described in the next subsection. We also note that our homogeneity condition is similar to condition (2.3)(i)(ii) of \cite{glynn1990simulation}. The difference is that the condition in \cite{glynn1990simulation} is on the estimator of variability, while our condition is on the end points. Correspondingly, \cite{glynn1990simulation} requires that the variability estimator is translation invariant, i.e.,  $g(\mathbf{x}+c1_k)=g(\mathbf{x})$, while our condition 2 requires that the end points translates in a way the same as the Stage 1 estimates. Moreover, in terms of usage, \cite{glynn1990simulation} proposes their condition to construct STS schemes, which is different from our motivation to handle alternative hypotheses needed in analyzing length optimality among general CIs under computation budgets.

\subsection{Asymptotic UMA Unbiasedness and Its Connection to Shortest Interval Length}

Constructing a CI for $\psi(P)$ can be viewed as the dual of the hypothesis test where the null hypothesis is $\psi(\tilde P)=\psi(P)$ (where $\tilde P$ here denotes the unknown distribution in the test). In our problem, the CI will shrink to a singleton at $\psi(P)$ at speed $n^{1/2}$ as $n$ increases. Correspondingly, we consider the hypothesis test with null hypothesis $\delta = 0$ where $\psi(\tilde P) = \psi(P) + n^{-1/2}\delta$. To proceed, we consider the well-established criterion of UMP unbiasedness in hypothesis testing. A test is called unbiased if for every alternative distribution, its power is greater than the size (the upper limit for the rejection probability when null hypothesis holds), and a test is UMP unbiased if it maximizes the power at every alternative distribution among all unbiased tests. 
From these, we can provide an asymptotic generalization and state the optimality of CI as being the dual of a UMP unbiased test, as detailed in Definition \ref{def: AUMAU}.

\begin{definition}\label{def: AUMAU}
\textit{(Asymptotically UMA unbiased CIs)} Given data sequence
$\{\mathbf{Y}_{n},n=1,2,\dots\}$, a CI represented as a pair of functions
$C(y)=[L(y),U(y)]$ is \emph{asymptotically unbiased} at level $1-\alpha$
if 
\[
{\lim}_{n\to\infty}P(\psi(P)\in C(\mathbf{Y}_{n}))\geq1-\alpha
\]
 and for any $\delta\neq0$, 
\[
\limsup_{n\to\infty}P(\psi(P)+n^{-1/2}\delta\in C(\mathbf{Y}_{n}))\leq1-\alpha.
\]
 $C$ is \emph{asymptotically uniformly most accurate (UMA) unbiased} at level
$1-\alpha$ among class $\mathcal{C}$ if it minimizes the probabilities
\[
\limsup_{n\to\infty}P(\psi(P)+n^{-1/2}\delta\in C(\mathbf{Y}_{n}))
\]
 among all asymptotically unbiased level $1-\alpha$ tests $C$ in
$\mathcal{C}$ for any $\delta\neq0$.
\end{definition}

The restriction to unbiasedness is necessary when defining the optimality for two-sided CIs/hypothesis tests, as discussed in \cite{lehmann2005testing}. Otherwise, when compared with some one-sided CIs which have smaller coverage probabilities on one side, two-sided CIs cannot have uniformly smaller coverage. Moreover, we can check that a general class of two-sided CIs where the point estimator satisfies a CLT whose limiting distribution is independent of the half width is unbiased, as shown in the following lemma. In particular, this means that all the methods we described earlier, including batching, overlapping batching, and cheap bootstrap are unbiased. Hereafter, we shorthand $a\pm b$ as an interval $[a-b,a+b]$ when no confusion arises.

\begin{lemma}\label{lem: unbiased example}
    Suppose that a CI is given by $C = \psi_{n}\pm n^{-1/2}qA_n$
where $(\sqrt{n}(\psi_{n}-\psi(P)),A_n)$ converges in distribution to $(Z,A)$ such that $Z$ is normal and independent of $A$ under $P$. Moreover, suppose that $C$ is asymptotically exact at level $1-\alpha$, i.e., 
\[
{\lim}_{n\to\infty}P(\psi(P)\in C)=1-\alpha.
\]
Then, $C$ is asymptotically unbiased at level $1-\alpha$.
\end{lemma}

To prove Lemma \ref{lem: unbiased example}, we argue that the asymptotic distribution of the pivotal statistic has higher density around its center, so the maximum coverage is achieved when the test point is asymptotically at the center of the CI. A concrete proof is provided in Appendix \ref{subsec: proof}. 

Compared with the classical definition of UMA/UMP unbiasedness, one difference in our definition is that we are looking at probabilities only under $P$. In contrast, the classical statement would introduce alternative distributions and study the coverage/rejection probabilities under different distributions. The reason that we can get rid of this is our homogeneity conditions. For example, the probabilities $P(\psi(P)+n^{-1/2}\delta\in C(\mathbf{Y}_{n}))$ can also be viewed as the probability of $\psi(P)+n^{-1/2}\delta\in C(\mathbf{Y}_{n})$ under a local alternative distribution $\tilde P$ where $\psi(\tilde P)=\psi(P)+n^{-1/2}\delta$. While introducing these local alternatives is potentially feasible, it involves tremendous technicality that can be overcome by using our homogeneity conditions on the structure of the CI itself.

Lastly, we show that asymptotic UMA unbiasedness for a CI implies the shortest asymptotic expected length among unbiased CIs in the same class, under slightly stronger conditions on the moments of the length. More precisely, we have the following:
\begin{theorem}\label{thm: UMA to length}
\textit{(From asymptotic UMA unbiasedness to shortest asymptotic expected length)}
    Given data sequence
$\{\mathbf{Y}_{n},n=1,2,\dots\}$, suppose that $C(y)=[L(y),U(y)]$ is asymptotically UMA unbiased at level $1-\alpha$ among class $\mathcal{C}$. Suppose further that there exists $\epsilon>0$ and constant $M$ such that both $\mathbb{E}\left\vert\sqrt{n}(U(\mathbf{Y}_n)-\psi(P))_{+}\right\vert^{1+\epsilon}\leq M$ and $\mathbb{E}\left\vert\sqrt{n}(\psi(P)-L(\mathbf{Y}_n))_{+}\right\vert^{1+\epsilon}\leq M$ for any $n\geq 1$. Moreover, suppose that $\lim_{n\to\infty}P(\psi(P)+n^{-1/2}\delta\in C^{\prime}(\mathbf{Y}_{n}))$ exists for any $C^{\prime}\in\mathcal{C}$. Then, for any asymptotically unbiased CI $C_1=[L_1(y),U_1(y)]$ in $\mathcal{C}$  with level $1-\alpha$, we have that 
\[
\liminf_{n\to\infty}\frac{ \mathbb{E}[U_1(\mathbf{Y}_{n})-L_1(\mathbf{Y}_{n})]}{\mathbb{E}[U(\mathbf{Y}_{n})-L(\mathbf{Y}_{n})]}\geq 1
\]
\end{theorem}

 The proof of Theorem \ref{thm: UMA to length} is based on an integration argument which goes from the optimality in terms of coverage to the optimality in terms of length. In Theorem \ref{thm: UMA to length}, we assume the uniform boundedness of the $1+\epsilon$ moments of $\{\sqrt{n}(U(\mathbf{Y}_n)-\psi(P))_{+}\}_{n=1,2,\dots}$ and $\{\sqrt{n}(U(\mathbf{Y}_n)-\psi(P))_{+}\}_{n=1,2,\dots}$. This assumption can be weakened to uniform integrability when we have more information on the form of the CI and the CLT for the Stage 1 estimates, so that the conditions needed become more transparent and verifiable, as will be seen in the next section. Nonetheless, Theorem \ref{thm: UMA to length} highlights the important general implication in translating asymptotic UMA unbiasedness to shortest asymptotic interval length. Despite the long-standing literature in UMA, this implication is new as far as we are aware of.

\section{Statistically Optimal CIs}\label{sec: AUMAU}

Recall that low-computation CI construction methods generally follow the two stages shown in Figure \ref{fig: two stages}. With this dissection, we will first find the optimal way to construct $L(\cdot),U(\cdot)$ in Stage 2, given the mechanism in producing Stage 1 estimates. This corresponds to \emph{local optimality} which we detail in Section \ref{subsec: optimal combination}. Then, we compare different mechanisms in Stage 1 to conclude the overall optimality among all homogeneous CIs. This corresponds to \emph{global optimality} which we detail in Section \ref{sec:global}.

\subsection{Local Optimality: Optimal Formula to Combine Stage 1 Estimates}\label{subsec: optimal combination}

To begin, we show that when data are divided into equal-size non-overlapping batches in Stage 1, the standard batching formula is optimal. 
\begin{theorem}
\label{thm: std batching}(Optimality of standard batching) Assume that $\sqrt{n}\left(\psi(\hat{P}_n)-\psi(P)\right)\Rightarrow N(0,\sigma^2)$ for some $\sigma^2>0$. Let $\mathcal{H}_B$ denote the class of homogeneous CIs using equal-size non-overlapping batch estimates $\mathbf{Y}_{n,B}$. 
Then, 
\begin{itemize}
    \item Standard batching given by $CI_{B}:=[L_{B}(y),U_{B}(y)]:=\bar{y}\pm t_{K-1;1-\alpha/2}S/\sqrt{K}$,
where $\bar{y}=\sum_{j=1}^{K}y_{j}/K$, $t_{K-1;1-\alpha/2}$ is the
$1-\alpha/2$ quantile of the $t$ distribution with $K-1$ degrees
of freedom, and $S^{2}=\sum_{j=1}^{K}\left(y_{j}-\bar{y}\right)^{2}/(K-1)$
is asymptotically UMA unbiased at level $1-\alpha$
among class $\mathcal{H}_B$. 

\item Assume further that $\left\{ \sqrt{n}\left(\psi\left(\hat{P}_n\right)-\psi\left(P\right)\right)\right\} _{n=0,1,\dots}$
is uniformly integrable. Then, for any asymptotically unbiased $CI$
in $\mathcal{H}_B$ with level $1-\alpha$, we have that $\liminf_{n\to\infty}\frac{ \mathbb{E}\left[\text{length of }CI\right]}{ \mathbb{E}[\text{length of }CI_{B}]}\geq1$.
\end{itemize}

\end{theorem}

The proof of Theorem \ref{thm: std batching} is provided in Appendix \ref{subsec: proof}. In the proof, we first use the CLT and homogeneity conditions to express the asymptotic coverage probabilities at $\psi(P)+n^{-1/2}\delta$ in terms of the probability of $L(\mathbf{Z})\le0\leq U(\mathbf{Z})$ under different normal distributions in the family $\{N(\mu1_{K},\sigma^{2}I):\mu\in R,\sigma^{2}>0\}$. Then, the problem is reduced to hypothesis testing for this distribution family. For this family, a UMP unbiased test is known, so we can directly convert it to its dual CI. Finally, to get the statement on the expected length, we use the integration argument as in Theorem \ref{thm: UMA to length}, but also the particular form of the batching CI so that we can relax the uniform boundedness of $1+\epsilon$ moments in Theorem \ref{thm: UMA to length} to uniform integrability.  

Next, we consider general Stage 1 schemes where we have estimates
$\mathbf{Y}_{n}=(\psi_{1},\dots,\psi_{K})$ that can be dependent.
We consider estimators that satisfy
$\sqrt{n}(\mathbf{Y}_{n}-\psi(P))\Rightarrow N(\mu1_{K},\sigma^{2}\Sigma)$ where $\sigma>0$ is unknown but $\Sigma$ is known. In other words, the asymptotic covariance is known up to a scaling. Typically, $\Sigma$ can be determined from the batch or resample sizes and amount of batch overlaps in Stage 1. For example, the following lemma tells us how to find $\Sigma$ for a class of overlapping batching schemes:

\begin{lemma}\label{lem: corr formula}
   Suppose that $X_1,\dots,X_n$ are i.i.d. generated from $P$. Suppose that the $j$-th batch has batch size $\gamma_jn$ and for $1\leq i\neq j\leq K$, batch $i$ and batch $j$ share $\beta_{ij}n$ samples. Moreover, suppose that each batch can be represented as $\cup_{i=1}^m \{X_{a_i n+1},X_{a_i n+2},\dots,X_{b_i n}\}$ for a sequence $0\leq a_1\leq b_1\leq a_2\leq \dots \leq a_m\leq b_m\leq 1$ that does not depend on $n$. Suppose further that there exists an influence function $G$ such that $\sqrt{n}\left(\psi(\hat{P}_n) - \psi(P) - \mathbb{E}_{\hat{P}_n}G(X)\right) \to 0$ in probability as $n\to \infty$ and $\mathbb{E}_{P}G(X)=0, 0<Var_{P}G(X)<\infty$. Then, there exists $\sigma>0$ such that $\sqrt{n}(\mathbf{Y}_{n,OB}-\psi(P)1_K)\Rightarrow N(0,\sigma^2 V)$ where $V_{ii}=1/\gamma_i$ and $V_{ij}=\frac{\beta_{ij}}{\gamma_i\gamma_j}$ for $i\neq j, 1\leq i,j\leq K$.
\end{lemma}

In Lemma \ref{lem: corr formula}, we introduced influence function $G$ which can be viewed as an analog of derivative in the space of probability distributions (\citealt{serfling1980}). With the help of $G$, for $Q$ that is close to $P$, we are able to approximate $\psi(Q)-\psi(P)$ with  $\mathbb{E}_{Q}G(X)$. This helps us analyze the behavior of $\psi(\cdot)$ around $P$.  

Lemma \ref{lem: corr formula} can also be generalized to the case where the data is a stationary sequence. Following the literature on CLT in the dependent case, we introduce the mixing condition.

\begin{definition}
\textit{(Mixing conditions)}
    For a sequence $X_1,X_2,\dots$ of random variables, let $\alpha_n$ be a number such that 
    \[
    \left\vert P(A\cap B) - P(A)P(B)\right\vert \leq \alpha_n
    \]
    for $A\in\sigma(X_1,\dots,X_k)$,  $B\in \sigma(X_{k+n},X_{k+n+1},\dots)$, and $k\geq 1, n\geq 1$. We say that $\{X_n\}$ is $\alpha$-mixing if $\alpha_n\to 0$ as $n\to\infty$. We say that $\{X_n\}$ is stationary if the distribution of $(X_{n},X_{n+1},\dots,X_{n+j})$ does not depend on $n$.
\end{definition}

For example, if $Y_n = f(X_n)$ where $\{X_n\}$ is a stationary discrete state-space Markov chain, then it is shown in Example 27.6 of \cite{billingsley1995probability} that $\{Y_n\}$ is $\alpha$-mixing with rate $\alpha_n = s\rho^n$ where $s$ is the number of states and $\rho<1$.

We can show the following extension of Lemma \ref{lem: corr formula} for dependent data. This is derived from a similar proof, but using the CLT for mixing sequences as in Theorem 27.4 of \cite{billingsley1995probability} to show the limit of $n^{-1/2}\sum_{i=1}^nG(X_i)$. 

\begin{lemma}\label{lem: CLT dependent}
    Under the same conditions as Lemma \ref{lem: corr formula} but without assuming that $\{X_1,\dots,X_n\}$ are i.i.d., suppose further that $\{G(X_n)\}_{n\geq 1}$ ($G$ is the influence function introduced as in Lemma \ref{lem: corr formula}) is stationary and $\alpha$-mixing with mixing coefficient $\alpha_n=O(n^{-5})$. Then, the same result as Lemma \ref{lem: corr formula} holds.
\end{lemma}

For any given $\Sigma$ induced from stage 1, we can find the optimal CI as detailed in Theorem \ref{thm: general batching} below.

\begin{theorem}
\label{thm: general batching}(Optimality using general Stage 1 estimates) Let $\mathcal{H}_{GS}^{(\Sigma)}$ denote the class of homogeneous CI and Stage 1 estimates $\mathbf{Y}_{n,GS}=(\psi_{1},\dots,\psi_{K})$ that satisfy  $\sqrt{n}(\mathbf{Y}_{n,GS}-\psi(P))\Rightarrow N(0,\sigma^{2}\Sigma)$ for some $\sigma>0$.
Then,
\begin{itemize}
    \item The following CI is asymptotically UMA unbiased
at level $1-\alpha$ among class $\mathcal{H}_{GS}^{(\Sigma)}$:
\[
CI_{GS}^{(\Sigma)}\left(\mathbf{Y}_{n,GS}\right):=\frac{1_{K}^{T}\Sigma^{-1}\mathbf{Y}_{n,GS}}{\lambda}\pm\frac{t_{K-1;1-\alpha/2}}{\sqrt{\lambda(K-1)}}\sqrt{\left(\mathbf{Y}_{n,GS}-\frac{1_{K}^{\top}\Sigma^{-1}\mathbf{Y}_{n,GS}}{\lambda}1_{K}\right)^{\top}\Sigma^{-1}\left(\mathbf{Y}_{n,GS}-\frac{1_{K}^{\top}\Sigma^{-1}\mathbf{Y}_{n,GS}}{\lambda}1_{K}\right)}
\]
Here, $\lambda=1_{K}^{\top}\Sigma^{-1}1_{K}$.

\item  Assume further that
$\left\{ \sqrt{n}\left(\psi\left(\hat{P}_n\right)-\psi\left(P\right)\right)\right\} _{n=0,1,\dots}$
is uniformly integrable. Then, for any asymptotically unbiased $CI$
in $\mathcal{H}_{GS}^{(\Sigma)}$ with level $1-\alpha$, we have that $\liminf_{n\to\infty} \frac{\mathbb{E}\left[\text{length of }CI\right]}{\mathbb{E}[\text{length of }CI_{GS}^{(\Sigma)}]}\geq1$.
\end{itemize}

\end{theorem}

Note that Theorem \ref{thm: general batching} does not require that the original input data $X_1,X_2,\dots$ are i.i.d., nor does it confine to deterministic or random selection of data to create Stage 1 estimates - it holds as long as the CLT $\sqrt{n}(\mathbf{Y}_{n,GS}-\psi(P))\Rightarrow N(0,\sigma^{2}\Sigma)$ is valid. Theorem \ref{thm: general batching} thus applies very generally, including the more concrete examples below.

The argument for Theorem \ref{thm: general batching} uses a similar argument as that for Theorem \ref{thm: std batching}, namely that we deduce a hypothesis test for family $\{N(\mu1_{K},\sigma^{2}\Sigma):\mu\in R,\sigma^{2}>0\}$. Unlike Theorem \ref{thm: std batching}, however, the UMP unbiased test for this family does not appear known in the literature (note that it does not follow from a simple linear mapping such as $X \to \Sigma^{-1/2}X$ based on the test for $\{N(\mu1_{K},\sigma^{2}I):\mu\in R,\sigma^{2}>0\}$, since this mapping will not preserve the structure that the expectations of all the coordinate are the same). Therefore, we derive a UMP unbiased test from more fundamental theorems about exponential families, including the completeness or Basu's lemma, which can help obtain independence relations that are important for simplying the problem and identifying the limiting distribution of test statistics (see Theorem \ref{thm: UMP normal} and its proof in Appendix \ref{subsec: UMP normal} for more details). The detailed proof to Theorem \ref{thm: general batching} is provided in Appendix \ref{subsec: proof}.

\begin{example}[non-overlapping batching with general batch sizes]
    Suppose that the $j$-th batch has $\gamma_jn$ samples (assume for simplicity that $\gamma_jn,j=1,2,\dots$ are integers) where $\gamma_1+\gamma_2+\dots+\gamma_K=1$. Correspondingly, the batch estimates are given by $\psi_j^{(\gamma)}=\psi(\hat{P}_{\gamma_jn})$ where $\hat{P}_{\gamma_jn}$ is the empirical distribution of $X_{\left(\sum_{i=1}^{j-1}\gamma_in\right)+1},\dots, X_{\sum_{i=1}^j\gamma_in}$
From Theorem \ref{thm: general batching}, when the batch estimates are given by $\mathbf{Y}_{n,B}^{(\gamma)}=(\psi_1^{(\gamma)},\dots,\psi_K^{(\gamma)})^\top$, an asymptotically UMA unbiased CI at level $1-\alpha$ is given by 
\[
CI_B^{(\gamma)}:=\bar{\psi}^{(\gamma)} \pm t_{K-1;1-\alpha/2}\sqrt{\sum_{j=1}^{K}\gamma_j(\psi_j^{(\gamma)}-\bar{\psi}^{(\gamma)})^2}/\sqrt{K-1}
\]
where $\bar{\psi}^{(\gamma)}=\sum_{j=1}^K\gamma_j\psi_j^{(\gamma)}$. Moreover, under uniform integrability conditions, $CI_B^{(\gamma)}$ is also asymptotically shortest among unbiased CIs using  $\mathbf{Y}_{n,B}^{(\gamma)}$ in Stage 1.

\end{example}

\begin{example}[batched jackknife]
We use $\mathbf{Y}_{n,SJ}=(\psi_{1,SJ},\dots,\psi_{K,SJ})$ where each $\psi_{j,SJ}$ is evaluated on leave-one-batch-out data $\{X_1,\dots,X_n\}\setminus\{X_{(j-1)n/K+1},\dots,X_{jn/K}\}$. The construction in Section III.5b of \cite{asmussen2007} gives \[
CI_{BJ}=\bar J \pm \frac{t_{K-1;1-\alpha/2}}{\sqrt{K}}S
\]
where $J_i:=\sum_{j=1}^K \psi_{j,SJ} - (K-1) \psi_{i,SJ}$ and $S^2=\sum_{i=1}^K(J_i-\bar J)^2/(K-1)$. We can check that this equals the construction in Theorem \ref{thm: general batching}. Indeed, we have that $\gamma_j=(K-1)/K$ and $\beta_{ij}=(K-2)/K$. Therefore, using Lemma \ref{lem: corr formula}, we have that the asymptotic covariance of $\mathbf{Y}_{n,SJ}$ is proportional to $V_{SJ} := \frac{I+(K-2)1_K1_K^{\top}}{K}$, and it can be computed that 
 $V^{-1}_{SJ}=K\left(I-\frac{K-2}{(K-1)^2}1_K1_K^{\top}\right)$ and $\lambda = \frac{K^2}{(K-1)^2}$. Based on these, it can be checked that $CI_{GS}^{(V_{SJ})}$ derived in Theorem \ref{thm: general batching} has point estimator $\frac{1_{K}^{T}V^{-1}\mathbf{Y}_{n,OB}}{\lambda} = \frac{1_K^T\mathbf{Y}_{n,OB}}{K}$ and variance estimator
\begin{align}
 & \frac{1}{\lambda}\left(\mathbf{Y}_{n,OB}-\frac{1_{K}^{\top}V^{-1}\mathbf{Y}_{n,OB}}{\lambda}1_{K}\right)^{\top}V^{-1}\left(\mathbf{Y}_{n,OB}-\frac{1_{K}^{\top}V^{-1}\mathbf{Y}_{n,OB}}{\lambda}1_{K}\right)\nonumber\\
= & \frac{1}{K}\left(\left(1_{K}1_{K}^{T}-(K-1)I-\frac{1}{K}1_{K}1_{K}^{\top}\right)\mathbf{Y}_{n,OB}\right)^{\top}\left(1_{K}1_{K}^{T}-(K-1)I-\frac{1}{K}1_{K}1_{K}^{\top}\right)\mathbf{Y}_{n,OB}\label{eq: var sj}
\end{align}
A proof of \eqref{eq: var sj} is provided in Appendix \ref{app: dev var}. Note that $\frac{1}{K}1_K1_K^{\top}\mathbf{Y}_{n,OB}$ is the average of the batch estimates and the $i$-th coordinate of $\left(1_K1_K^\top-(K-1)I\right)\mathbf{Y}_{OB}$ is given by  $\sum_{j=1}^K\psi_{j,SJ}-(K-1)\psi_{i,SJ}$. Therefore, comparing with $CI_{BJ}$, we get the claim.
\end{example}

\begin{example}[overlapping batching]
Recall the construction of overlapping batches in \cite{su2023overlapping} which gives Stage 1 estimators
$\mathbf{Y}_{n,OB}=\left(\psi_{j,OB},1\leq j\leq K\right)$. Here, the first batch is the entire sample $X_1,\ldots,X_n$. For $2\leq j\leq K$, batch $j$ starts from the $(j-2)\frac{n-\gamma n}{K-2}+1$-th and ends at the $(j-2)\frac{n-\gamma n}{K-2}+\gamma n$-th data point. For this Stage 1, following  Theorem \ref{thm: general batching}, we propose the UMA unbiased CI $CI_{OB\text{-new}}:=CI_{GS}^{(V_{OB})}$, i.e., 
\[
\frac{1_{K}^{T}V_{OB}^{-1}\mathbf{Y}_{n,OB}}{\lambda}\pm\frac{t_{K-1;1-\alpha/2}}{\sqrt{\lambda(K-1)}}\sqrt{\left(\mathbf{Y}_{n,OB}-\frac{1_{K}^{\top}V_{OB}^{-1}\mathbf{Y}_{n,OB}}{\lambda}1_{K}\right)^{\top}V_{OB}^{-1}\left(\mathbf{Y}_{n,OB}-\frac{1_{K}^{\top}V_{OB}^{-1}\mathbf{Y}_{n,OB}}{\lambda}1_{K}\right)}
\]
Here, $\lambda=1_{K}^{\top}V_{OB}^{-1}1_{K}$ and $V_{OB}$ is the matrix from Lemma \ref{lem: corr formula} where $\gamma_i=\gamma$ and $\beta_{ij} = \left(\gamma - |i-j|\frac{1-\gamma }{K-2}\right)_+$ for $2\leq i,j\leq K$, $\gamma_{1}=1$, and $\beta_{i1}=\gamma$ for $i=1\leq i\leq K$. We note that our $CI_{OB\text{-new}}$ and $CI_{OB}$ in \cite{su2023overlapping} use the same Stage 1 estimators, but put them together in different ways in Stage 2. Indeed, $CI_{OB\text{-new}}$ is calibrated based on a pivotal statistic with an asymptotic $t_{K-1}$ distribution, while $CI_{OB}$ is calibrated based on a more general distribution (which they denote by OB-I). Therefore, from the optimality of $CI_{OB\text{-new}}$ shown in Theorem \ref{thm: general batching}, we expect it to be shorter than $CI_{OB}$.
\end{example}

\begin{example}[cheap bootstrap]
Under regularity conditions, it is shown in Proposition 1 of \cite{lam2022cheap} that the following joint CLT holds:
\[
\left(\sqrt{n}(\psi(\hat{P}_n)-\psi(P)),\sqrt{n}(\psi(P^{*1}_n)-\psi(\hat{P}_n)),\dots,\sqrt{n}(\psi(P^{*(K-1)}_n)-\psi(\hat{P}_n))\right)^\top\Rightarrow N(0,\sigma_0^2 I).
\]
Therefore, for $\mathbf{Y}_{n,CB}=(\psi(\hat{P}_n),\psi(P^{*1}_n),\dots,\psi(P^{*(K-1)}_n))$, we have that
\[
\sqrt{n}(\mathbf{Y}_{n,CB}-\psi(P))\Rightarrow N(0,\sigma_0^2 V).
\]
where $V_{11}=1,V_{1i}=1,V_{ii}=2,V_{ij}=1$ for $2\leq i\neq j\leq K$ and it can be calculated that 
\[
V^{-1}=\left[\begin{array}{ccccc}
K & -1 & -1 & \dots & -1\\
-1 & 1 & 0 & \dots & 0\\
-1 & 0 & 1 & \dots & 0\\
\dots & \dots & \dots & \dots & \dots\\
-1 & 0 & 0 & \dots & 1
\end{array}\right]
\]
From this, it can be calculated that the CI derived in Theorem \ref{thm: general batching} corresponding to $V$ has point estimator $\frac{1_{K}^{T}V^{-1}\mathbf{Y}_{n,CB}}{\lambda} = \psi(\hat{P}_n)$ and variance estimator
\begin{align*}
 & \frac{1}{\lambda}\left(\mathbf{Y}_{n,CB}-\frac{1_{K}^{\top}V^{-1}\mathbf{Y}_{n,CB}}{\lambda}1_{K}\right)^{\top}V^{-1}\left(\mathbf{Y}_{n,CB}-\frac{1_{K}^{\top}V^{-1}\mathbf{Y}_{n,CB}}{\lambda}1_{K}\right)\nonumber\\
= & \sum_{j=1}^{K-1}(\psi(P^{*j}_n)-\psi(\hat{P}_n))^2
\end{align*}
which is the same as $CI_{CB}$. Therefore, from Theorem \ref{thm: general batching}, we have that $CI_{CB}$ is asymptotically UMA unbiased given data $\mathbf{Y}_{n,CB}$. Moreover, under uniform integrability conditions it has the shortest expected length among unbiased CIs using data $\mathbf{Y}_{n,CB}$.
\end{example}

\begin{example}[weighted cheap bootstrap]
From Appendix \ref{subsec: CLT weighted}, we have the following joint CLT
\begin{align*}
 & \left(\sqrt{n}(\psi(\hat{P}_{n})-\psi(P)),\sqrt{n}(\psi(P_{n}^{W,1})-\psi(\hat{P}_{n})),\dots,\sqrt{n}(\psi(P_{n}^{W,K-1})-\psi(\hat{P}_{n}))\right)^{\top}\\
\Rightarrow & N(0,\sigma_0^{2}diag(1,\sigma_{W}^{2},\sigma_{W}^{2},\dots,\sigma_{W}^{2})).
\end{align*}
Here, $diag(x_1,\dots,x_K)$ means a matrix whose diagonal terms are given by $x_1,\dots,x_K$ and other terms are zero. Therefore, for $\mathbf{Y}_{n,WCB}=(\psi(\hat{P}_n),\psi(P^{W,1}_n),\dots,\psi(P^{W,K-1}_n))$, we have that
\[
\sqrt{n}(\mathbf{Y}_{n,WCB}-\psi(P)) \Rightarrow N(0,\sigma_0^2 V).
\]
where $V_{11}=1,V_{1i}=1,V_{ii}=1+\sigma_W^2,V_{ij}=1$ for $2\leq i\neq j\leq K$ and it can be calculated that 
\[
V^{-1}=\left[\begin{array}{ccccc}
1+(K-1)\sigma_W^{-2} & -\sigma_W^{-2} & -\sigma_W^{-2} & \dots & -\sigma_W^{-2}\\
-\sigma_W^{-2} & \sigma_W^{-2} & 0 & \dots & 0\\
-\sigma_W^{-2} & 0 & \sigma_W^{-2} & \dots & 0\\
\dots & \dots & \dots & \dots & \dots\\
-\sigma_W^{-2} & 0 & 0 & \dots & \sigma_W^{-2}
\end{array}\right]
\]
Similar to the cheap bootstrap, we can check that $\frac{1_{K}^{T}V^{-1}\mathbf{Y}_{n,WCB}}{\lambda} = \psi(\hat{P}_n)$ and 
\begin{align*}
 & \frac{1}{\lambda}\left(\mathbf{Y}_{n,WCB}-\frac{1_{K}^{\top}V^{-1}\mathbf{Y}_{n,WCB}}{\lambda}1_{K}\right)^{\top}V^{-1}\left(\mathbf{Y}_{n,WCB}-\frac{1_{K}^{\top}V^{-1}\mathbf{Y}_{n,WCB}}{\lambda}1_{K}\right)\nonumber\\
= & \sigma_W^{-2}\sum_{j=1}^{K-1}(\psi(P^{W,j}_n)-\psi(\hat{P}_n))^2
\end{align*}
which is the same as $CI_{WCB}$. Therefore, like in the original cheap bootstrap case, from Theorem \ref{thm: general batching} we have that $CI_{WCB}$ is asymptotically UMA unbiased given data $\mathbf{Y}_{n,WCB}$ and, moreover, under uniform integrability conditions it has the shortest expected length among unbiased CIs using data $\mathbf{Y}_{n,WCB}$.
\end{example}

\subsection{Global Optimality: Asymptotically Globally Shortest Unbiased CIs}\label{sec:global}

Now we study the optimality of CIs among different Stage 1 mechanisms. We focus on those using Stage 1 estimates that satisfy a joint CLT: Let $\mathcal{H}_{GS}$ denote the class of CI using $[L(\cdot),U(\cdot)]\in\mathcal{H}$ and Stage 1 estimates $\mathbf{Y}_n$ that satisfies $\sqrt{n}(\mathbf{Y}_n-\psi(P))\Rightarrow N(0,\Sigma)$ for some $\Sigma$. 

\begin{definition}
\textit{(Asymptotically globally shortest unbiased CIs)}
    We say that $C=[L(\mathbf{Y}_n),U(\mathbf{Y}_n)]$  is \textit{asymptotically globally shortest unbiased} in class $\mathcal{H}_{GS}$ if it is asymptotically unbiased and for any asymptotically unbiased CI $C^\prime=[L^\prime(\mathbf{Y}^\prime_n),U^\prime(\mathbf{Y}^\prime_n)]$ in $\mathcal{H}_{GS}$, we have that 
\[
\liminf_{n\to\infty} \mathbb{E}[\text{length of }C^\prime]/ \mathbb{E}[\text{length of }C]\geq 1
\]
\end{definition}

The following result provides a characterization of CIs with the global shortest asymptotic expected length in terms of the limits of the affine combinations of Stage 1 estimates.

\begin{theorem}\label{thm: global}
\textit{(Characterization of asymptotically globally shortest unbiased CIs)}
    Assume that $\left\{ \sqrt{n}\left(\psi\left(\hat{P}_n\right)-\psi\left(P\right)\right)\right\} _{n=0,1,\dots}$
is uniformly integrable, then $CI_{GS}^{(\Sigma)}(\mathbf{Y}_n)$ constructed in Theorem \ref{thm: general batching} is asymptotically globally shortest unbiased in $\mathcal{H}_{GS}$ if and only if there exists $w\in\mathbb{R}^K$ such that $1^\top w=1$ and $\sqrt{n}(w^\top\mathbf{Y}_n-\psi(P))\Rightarrow N(0,\sigma_0)$. Here, recall that $\sigma_0^2$ is the asymptotic varaince of $\psi(\hat{P}_n)$, i.e., $\sqrt{n}\left(\psi(\hat{P}_n)-\psi(P)\right)\Rightarrow N(0,\sigma_0^2)$. 
\end{theorem}

Theorem \ref{thm: global} is derived by finding the condition for the asymptotic expected length of $CI_{GS}^{(\Sigma)}$ to achieve its minimum as $\Sigma$ changes. The detailed proof is provided in Appendix \ref{subsec: proof}. Theorem \ref{thm: global} implies that $CI_{GS}^{(\Sigma)}(\mathbf{Y}_n)$ is globally shortest unbiased if $\psi(\hat{P}_n)$ can be asymptotically represented as an affine combination of $\mathbf{Y}_n$ in the sense that $\sqrt{n}(\psi(\hat{P}_n)-w^\top \mathbf{Y}_n)\stackrel{p}{\to} 0$ for some $w$ such that $1^\top w=1$. 
From this observation, we can study the global optimality for the CIs in Section \ref{subsec: optimal combination}, as below. 

\begin{corollary}\label{cor: equivalence}
\textit{(Equivalence of CIs in attaining global optimality)}    
Under the same conditions as Theorem 
\ref{thm: global}, all of $CI_B,CI_B^{(\gamma)},CI_{BJ},CI_{CB}$, and $CI_{OB}$ are asymptotically globally shortest unbiased in $\mathcal{H}_{GS}$. Moreover, all of these CIs are asymptotically equivalent, in the sense that 1) the center of all these CIs differs from $\psi(\hat{P}_n)$ by $o(n^{-1/2})$ and 2) the asymptotic distribution of $\sqrt{n}$ times the length of these CIs are all equal to $\sigma_0\sqrt{\frac{\chi^2_{K-1}}{K-1}}$ where $\sigma_0$ is the asymptotic variance of $\psi(\hat{P}_n)$. 
\end{corollary}

The proof of Corollary \ref{cor: equivalence} is also provided in Appendix \ref{subsec: proof}. Corollary \ref{cor: equivalence} stipulates that all of batching, batched jackknife, cheap bootstrap, and batching with general or overlapping batches proposed in this paper are globally optimal and moreover, they give asymptotically equivalent $t$-based CIs. 
Finally, it is shown in \cite{glynn1990simulation} that among the class of CIs with the same pivotal statistic, the one that is symmetric around the empirical estimator has the shortest expected length. Combining that result with our findings in Theorem \ref{thm: global}, we can further show the optimality among the class of CIs that is not necessarily unbiased. Detailed discussions are provided in Section \ref{subsec: biased CI}.

\section{Numerical Experiments}\label{sec: numerics}

In this section, we provide numerical experiments to test the empirical performance of the low-computation CIs proposed in this paper. The CIs we consider are \begin{enumerate}
\item Standard batching: $CI_{B}$ given in Section \ref{sec: setup}.
\item Batching with general batch sizes: $CI_{B}^{(\gamma)}$ as introduced
in Example 1. We use $\gamma_{j}=\frac{j}{K(K+1)/2},1\leq j \leq K$.
\item Cheap bootstrap: $CI_{CB}$ given in Section \ref{sec: setup}.
\item Overlapping batching that we proposed: $CI_{OB\text{-new}}$ as introduced in Example 3 where we set $\gamma$ as $0.3$. Note that $\gamma$ measures the batch size (the
size of each batch is $\gamma n$). Our choice of the batch size is similar to the size used in Section 9 of \cite{su2023overlapping} which considered $\gamma=0.1,0.25$ (our choice of $\gamma$ is a bit larger so that there is a sufficient amount of overlapping when the computational budget is low).
\item Overlapping batching proposed by \cite{su2023overlapping}: $CI_{OB}$ as given in Section \ref{sec: setup}, where the batch estimates
are the same as our overlapping batching outlined in item 4 above.
The critical values $c_{\beta,b,\alpha_{1}}$ can be checked
from Figure 3 in \cite{su2023overlapping} and in correspondance with our settings in item
4, we use $\beta=0.3,b=K-1,\alpha_{1}=0.95$.
\end{enumerate}

In the following subsections, we will test these CIs with different output functionals, including quantile, logistic regression, and stochastic simulation.

\subsection{Quantile Estimation}
We set $\psi(P)$ as the 0.7 quantile of $X\sim P$ where $P$ is standard log-normal.  We set the sample size as $n=3000$,
nominal level $\alpha=90\%$, and try each method under
different computational budgets ($K=6,12,17$) measured by the number
of times we can compute $\psi$. 
For each method, we replicate $10^5$ times to find its empirical coverages and half lengths. The results are shown in Table \ref{tab: numerics}. We can check that the empirical coverages are all close to the nominal coverage: the largest gap is observed for $CI_{OB}$ when $K=17$ where there is an overcoverage of about 1.3\%. This suggests that the sample size is large enough for our asymptotic argument to be true. If we compare the half lengths of different methods, we find that 1) the half lengths of $CI_B, CI_{B}^{(\gamma)}$, and $CI_{CB}$ are very close, which agrees with our Theorem 3 and the optimality of cheap bootstrap discussed at the end of Section 6. 2) the half length of $CI_{OB\text{-new}}$ is longer than $CI_B, CI_{B}^{(\gamma)}$, and $CI_{CB}$. For example, when $K=6$, the half length of $CI_{OB\text{-new}}$ is 0.082 where the half lengths of $CI_B, CI_{B}^{(\gamma)}$, and $CI_{CB}$ are all 0.078. This difference can be explained by its higher coverage, which is a result of finite-sample errors as we discussed. 3) The half length of $CI_{OB}$ is even longer than $CI_{OB\text{-new}}$. For example, when $K=12$, the half length of $CI_{OB}$ is longer than $CI_{OB\text{-new}}$ even though its coverage is smaller than $CI_{OB\text{-new}}$. This is expected as the construction in \cite{su2023overlapping} does not have the optimality guarantee in terms of CI accuracy as we do.

In addition, we consider another quantile estimation example where the goal is to estimate the 0.7 quantile of the stationary distribution of the system time for an M/M/1 model with arrival rate 1 and service rate 2. The sample size is still $n=3000$ as in the previous example, but there is a warm-up period with 1000 samples (i.e., we generate the system times for the first 4000 customers, but only use the last 3000 samples for the CI construction). The other settings are the same as the last example, but we will not run cheap bootstrap since we have dependent data. The results are shown in Table \ref{tab: quantile dependent}. Again, we can check that all of the coverages are close to the nominal level, which shows that the methods are also correct in the dependent case. For the comparisions among different batching methods: 1) $CI_{OB}$ also have the largest empirical expected half length for each value of $K$ tried. 2) The half length of $CI_{OB\text{-new}}$ is not uniformly longer than $CI_{B}$ and $CI_{B}^{(\gamma)}$ which was observed in the last experiment. Indeed, the half length of $CI_{OB\text{-new}}$ is larger than $CI_{B}$ and $CI_{B}^{(\gamma)}$ when $K=6$. But when $K$ increases to 17, the half length of $CI_{OB\text{-new}}$ becomes smaller than $CI_{B}$ and $CI_{B}^{(\gamma)}$. This suggests that the comparison among $CI_{B},CI_{B}^{(\gamma)}$ and $CI_{OB\text{-new}}$ can depend on problem parameters.

\begin{table}
\centering
\begin{tabular}{c|c|c|c|c|c|c}
\multirow{2}{*}{method} & \multicolumn{2}{c|}{$K=6$} & \multicolumn{2}{c|}{$K=12$} & \multicolumn{2}{c}{$K=17$}\tabularnewline
\cline{2-7} \cline{3-7} \cline{4-7} \cline{5-7} \cline{6-7} \cline{7-7} 
 & coverage & half length & coverage & half length & coverage & half length\tabularnewline
\hline 
$CI_{B}$ & 90.0\% & 0.078 & 90.0\% & 0.071 & 90.0\% & 0.070\tabularnewline
$CI_{B}^{(\gamma)}$ & 90.0\% & 0.078 & 89.9\% & 0.071 & 90.0\% & 0.070\tabularnewline

$CI_{CB}$ & 89.4\% & 0.078 & 89.6\% & 0.072 & 89.9\% & 0.071\tabularnewline
$CI_{OB\text{-new}}$ & 90.5\% & 0.082 & 91.1\% & 0.076 & 91.3\% & 0.075\tabularnewline
$CI_{OB}$ & 90.7\% & 0.086 & 90.3\% & 0.080 & 90.3\% & 0.079\tabularnewline
\end{tabular}

\caption{Empirical coverages and half lengths for each batching method. The uncertainty of these empirical estimates (measured by the 95\% CI half length) is less than 0.2\% for the empirical coverage and is less than $3\times 10^{-4}$ for the estimation of the half length. }
\label{tab: numerics}
\end{table}

\begin{table}
\centering
\begin{tabular}{c|c|c|c|c|c|c}
\multirow{2}{*}{method} & \multicolumn{2}{c|}{$K=6$} & \multicolumn{2}{c|}{$K=12$} & \multicolumn{2}{c}{$K=17$}\tabularnewline
\cline{2-7} \cline{3-7} \cline{4-7} \cline{5-7} \cline{6-7} \cline{7-7} 
 & coverage & half length & coverage & half length & coverage & half length\tabularnewline
\hline 
$CI_{B}$ & 89.7\% & 0.120 & 89.7\% & 0.114 & 89.3\% & 0.115\tabularnewline
$CI_{B}^{(\gamma)}$ & 89.7\% & 0.116 & 89.7\% & 0.116 & 89.1\% & 0.115\tabularnewline
$CI_{OB\text{-new}}$ & 90.1\% & 0.123 & 90.4\% & 0.113 & 90.3\% & 0.110\tabularnewline
$CI_{OB}$ & 90.5\% & 0.131 & 90.3\% & 0.121 & 90.0\% & 0.120\tabularnewline
\end{tabular}

\caption{Empirical coverages and half lengths for each batching method for the queueing model. The uncertainty of these empirical estimates (measured by the 95\% CI half length) is less than 0.3\% for the empirical coverage and is less than $3\times 10^{-4}$ for the estimation of the half length. }
\label{tab: quantile dependent}
\end{table}

\subsection{Logistic Regression}

Consider the model where $Y\in\{0,1\}$ and $P(Y=1|X)=e^{\beta^\top X}/(1+e^{\beta^\top X})$. $X$ is multivariate normal with mean 0 and $Cov(X_i,X_j)=0.01*0.8^{\left\vert i-j\right\vert},1\leq i,j\leq 10$. Let $\beta = (\beta_1,\beta_2,\dots,\beta_{10})$ where $\beta_1=\beta_2=\beta_3=1,\beta_4=\beta_5=\beta_6=-1$, and $\beta_7=\beta_8=\beta_9=\beta_{10}=0$. Let $P$ denote the joint distribution of $(X,Y)$ and $\psi(P)$ be the value of $\beta$ found through logistic regression. This represents a computationally more expensive model, for which we have a stronger reason to use low-computation methods (we are not able to run models that are too complicated though, as we need to replicate a sufficient number of times to estimate the coverage and length). To reduce estimation error, we construct a 90\% CI for each of $\beta_i,1\leq i\leq 10$ (not a simultaneous CI), and then evaluate the coverage and length by taking average. We set the sample size as $N=5\times 10^5$, and replicate the procedure for 4000 times to get empirical coverage and lengths. The results are given in Table \ref{tab: logistic}. The observations are similar to the previous examples, where the coverages are close to the nominal level, and all of $CI_{B},CI_{B}^{(\gamma)},CI_{OB\text{-new}},$ and $CI_{CB}$ have shorter empirical length than $CI_{OB}$. The difference is that, when $K$ is larger, $CI_{B}$ and $CI_{B}^{(\gamma)}$ have some undercoverage issue (e.g., when $K=17$, their coverages are 85.8\% and 84.5\% respectively, which are below the nominal level by about 5\%). This suggests that the problem is highly non-linear so that the estimators are not accurate when the batch size is not large enough (note that the batch sizes of $CI_{B}$ and $CI_{B}^{(\gamma)}$ are proportional to $K^{-1}$). On the other hand, the other methods ($CI_{OB\text{-new}},CI_{CB},CI_{OB}$) do not suffer from this issue, since their batch sizes do not reduce when $K$ increases.

\begin{table}
\centering
\begin{tabular}{c|c|c|c|c|c|c}
\multirow{2}{*}{method} & \multicolumn{2}{c|}{$K=6$} & \multicolumn{2}{c|}{$K=12$} & \multicolumn{2}{c}{$K=17$}\tabularnewline
\cline{2-7} \cline{3-7} \cline{4-7} \cline{5-7} \cline{6-7} \cline{7-7} 
 & coverage & half length & coverage & half length & coverage & half length\tabularnewline
\hline 
$CI_{B}$ & 89.5\% & 0.108 & 87.9\% & 0.097 & 85.8\% & 0.092\tabularnewline
$CI_{B}^{(\gamma)}$ & 89.3\% & 0.107 & 86.9\% & 0.094 & 84.5\% & 0.089\tabularnewline
$CI_{CB}$ & 90.1\% & 0.111 & 89.9\% & 0.102 & 89.7\% & 0.099\tabularnewline
$CI_{OB\text{-new}}$ & 89.7\% & 0.110 & 89.7\% & 0.101 & 89.6\% & 0.099\tabularnewline
$CI_{OB}$ & 90.5\% & 0.121 & 89.9\% & 0.112 & 89.8\% & 0.111\tabularnewline
\end{tabular}

\caption{Empirical coverages and half lengths for each batching method for the logistic regression example. The uncertainty of these empirical estimates (measured by the 95\% CI half length) is less than 0.4\% for the empirical coverage and is less than $5\times 10^{-4}$ for the estimation of the half length. }
\label{tab: logistic}
\end{table}

\subsection{Stochastic Simulation}

We consider a computer network model similar to \cite{cheng1997sensitivity,lin2015single,lam2022subsampling,ll2023}. There are four nodes and four channels as shown in Figure \ref{fig: network}. There are 12 types of messages corresponding to each distinct (ordered) pair of nodes with independent Possion arrival with rates shown in Table \ref{tab: network}. The length of each message (which is equal to the amount of time to transfer the message in any channel) has exponential distribution with mean 1/100. Each channel can only transfer one message at the same time. We assume that the arrival rates are known but we do not know the true distribution (denoted by $P$ for this example) of the length of message. We construct a CI for $\psi(P)$ defined as the average system time of the first 10 messages using a sample of $N=200$ empirical observations of message length. To reduce Monte-Carlo noise, for each evaluation of $\psi(\cdot)$ we simulate the network for 1000 times and then take the averaged output. We replicate 1000 times to estimate the empirical coverage and length. The results are shown in Table \ref{tab: network emp}. Again, we observe that the empirical coverage probabilities are close to the nominal probability, and that $CI_{OB}$ has the largest empirical length. Among the other four methods, $CI_{OB\text{-new}}$ has slightly larger half length when $K=6,12$ and its coverage is also larger. This could be caused by the finite-sample error or Monte Carlo noise.

\begin{figure}[htbp]
\centering
\includegraphics[width=0.45\textwidth]{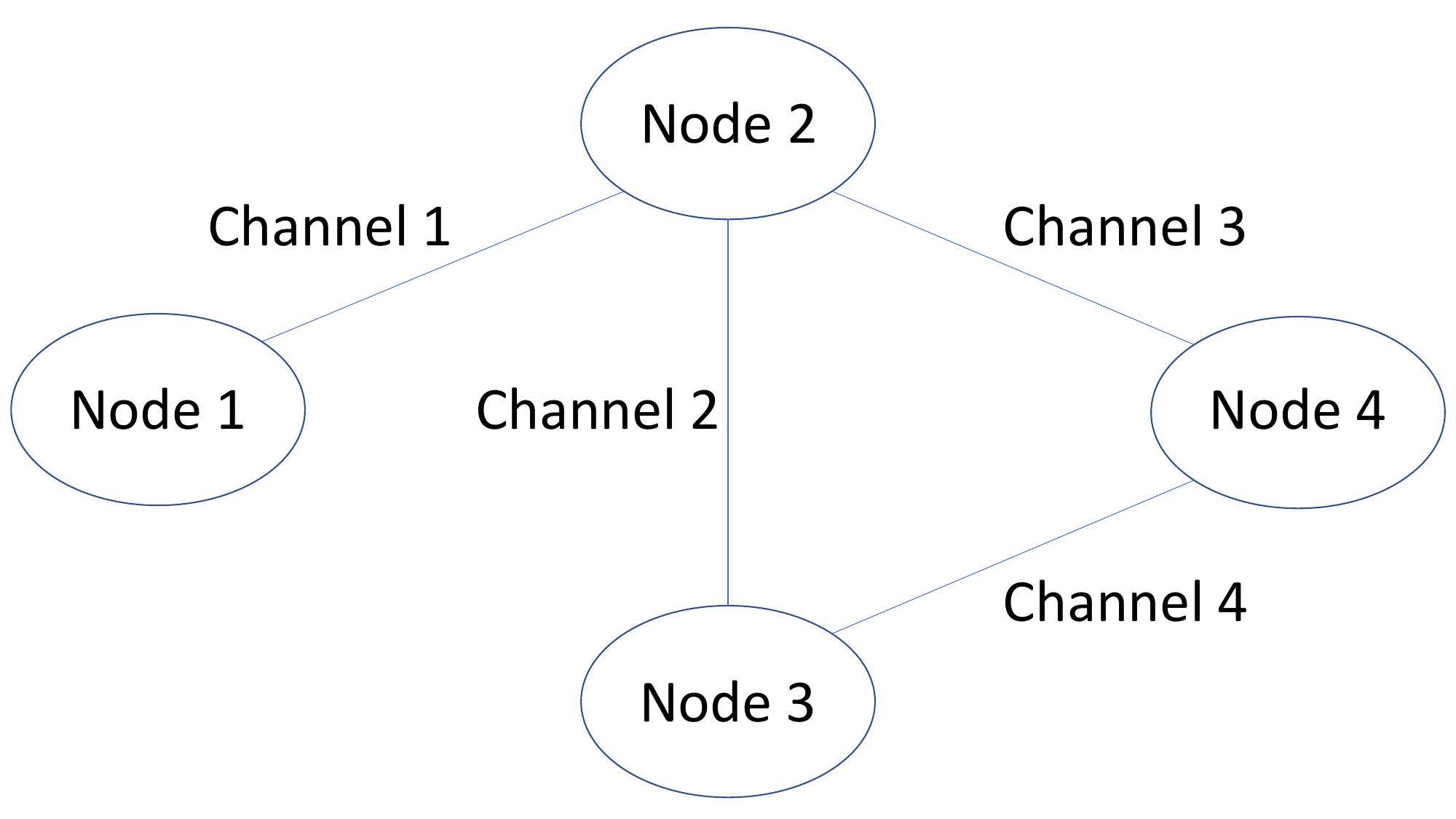}

\caption{A computer network with four nodes and four channels.}
\label{fig: network}
\end{figure}

\begin{table}
\centering
\begin{tabular}{|c|llll|}
\hline
              & \multicolumn{4}{c|}{Ending node}                                                 \\ \hline
Starting node & \multicolumn{1}{l|}{1}  & \multicolumn{1}{l|}{2}  & \multicolumn{1}{l|}{3}  & 4  \\ \hline
1             & \multicolumn{1}{l|}{NA} & \multicolumn{1}{l|}{40} & \multicolumn{1}{l|}{10} & 10 \\ \hline
2             & \multicolumn{1}{l|}{50} & \multicolumn{1}{l|}{NA} & \multicolumn{1}{l|}{45} & 15 \\ \hline
3             & \multicolumn{1}{l|}{10} & \multicolumn{1}{l|}{15} & \multicolumn{1}{l|}{NA} & 20 \\ \hline
4             & \multicolumn{1}{l|}{10} & \multicolumn{1}{l|}{30} & \multicolumn{1}{l|}{40} & NA \\ \hline
\end{tabular}
\caption{Arrival rates in the computer network.}
\label{tab: network}
\end{table}

\begin{table}
\centering
\begin{tabular}{c|c|c|c|c|c|c}
\multirow{2}{*}{method} & \multicolumn{2}{c|}{$K=6$} & \multicolumn{2}{c|}{$K=12$} & \multicolumn{2}{c}{$K=17$}\tabularnewline
\cline{2-7} \cline{3-7} \cline{4-7} \cline{5-7} \cline{6-7} \cline{7-7} 
 & coverage & half length & coverage & half length & coverage & half length\tabularnewline
\hline 
$CI_{B}$ & 90.2\% & 3.44E-3 & 88.9\% & 3.23E-3 & 89.1\% & 3.19E-3\tabularnewline
$CI_{B}^{(\gamma)}$ & 89.5\% & 3.48E-3 & 89.0\% & 3.21E-3 & 91.0\% & 3.27E-3\tabularnewline
$CI_{CB}$ & 89.3\% & 3.48E-3 & 88.9\% & 3.16E-3 & 88.9\% & 3.10E-3\tabularnewline
$CI_{OB\text{-new}}$ & 91.5\% & 3.62E-3 & 91.1\% & 3.33E-3 & 90.7\% & 3.25E-3\tabularnewline
$CI_{OB}$ & 92.3\% & 3.83E-3 & 90.2\% & 3.53E-3 & 90.0\% & 3.48E-3\tabularnewline
\end{tabular}

\caption{Empirical coverages and half lengths for each batching method for the stochastic simulation example. The uncertainty of these empirical estimates (measured by the 95\% CI half length) is less than 1.8\% for the empirical coverage and is less than $7\times 10^{-5}$ for the estimation of the half length. }
\label{tab: network emp}
\end{table}

\section{Conclusion and Future Works}\label{sec: future work}

In this paper, we study the construction of CIs for expensive black-box models that arise in large-scale estimation and simulation tasks. In this setting, we focus on low-computation CIs where we can only evaluate the target objective for a limited number of times. Under this computational budget, we create a framework to analyze the statistical optimality of CIs, in the sense of asymptotic interval shortness among the class of unbiased homogeneous CIs. Methodologically, we provide a systematic reduction of this problem into the attainment of asymptotic UMA intervals for Gaussian models, and and bridge the concept of asymptotic UMA with the minimization of the expected interval width. We show that standard batching, batching with uneven size, batched jackknife, the cheap bootstrap, and the weighted cheap bootstrap are all statistically optimal and asymptotically equivalent. Moreover, for overlapping batching, we propose a new formula that attains statistical optimality and gives asymptotically shorter CIs than methods proposed in the existing literature. 

We view this work as a foundation to initiate the study on the fundamental problem of uncertainty quantification for expensive computation models. There are several important immediate directions. One direction is to extend to the case where the outputs are multidimensional and the goal is to construct a confidence region for the output. If we follow the approach in this paper, the question would be equivalent as constructing a UMP hypothesis test where the null hypothesis involves multiple parameters. 

In addition, the setup in this paper corresponds to the case where the computation cost does not depend on the input data size. For example, the computation cost for evaluating $\psi(\hat{P}_n)$ and $\psi(\hat{P}_n^i)$ (recall that $\hat{P}_n^i,i=1,2,\dots,K$ are the batches in standard batching) are regarded as the same in our analysis. A more general setting is that the computation cost is affected jointly by the data size (in each batch/subsample) and the number of model evaluations. Finding the optimal CI under this  setting is another future direction.

Moreover, in this paper, we have assumed that there is no error in evaluating $\psi(\cdot)$, or that the error is negligible via the use of heavy computation per model run. As a generalization, one can also use less computation per run and allow for computational noise in evaluating $\psi(\cdot)$. The question is then what CI is optimal when we have both data noise and computation noise.

Lastly, we discuss potential alternatives to the restriction on the homogeneous class imposed in this paper. A more standard approach to describe the lack of knowledge on the true variability is to introduce an ambiguity set for all of the possible distributions (with different variability), and define the optimality in a minimax sense.  In the large-sample setting, to take care of the asymptotic scaling, the set of distributions should be replaced by a set of local alternative distribution sequences around the true distribution. For example, in the parametric setting, the local alternative distributions can be those with the true parameter perturbed by a small amount (say of order $n^{-1/2}$ where $n$ is the sample size). For the nonparametric setting considered in this paper, however, it is not clear how to explicitly construct these local alternative distribution sequences, since the form of the ambiguity set could depend on particular problem structures. Because of this, our approach in this paper is to introduce class homogeneity. Nonetheless, if we want to study other CIs that are not necessarily homogeneous, we may need to follow the more standard approach discussed above. In addition to the aforementioned challenge in constructing local alternative distributions, we also need to justify triangular-array CLTs under these distributions, which seems challenging especially when we study overlapping batching or the cheap bootstrap. This direction comprises another important line of future investigation.

\section*{Acknowledgements}
We gratefully acknowledge support from the National Science Foundation under grants CAREER CMMI-1834710 and IIS-1849280. A preliminary conference version of this work has appeared in \cite{he2023optimal}.

\bibliographystyle{informs2014}
\bibliography{references}

\newpage

\setcounter{equation}{0}
\setcounter{subsection}{0}
\setcounter{assumption}{0}
\setcounter{theorem}{1}
\setcounter{proposition}{0}
\renewcommand{\theequation}{A.\arabic{equation}}
\renewcommand{\thelemma}{A.\arabic{lemma}}
\renewcommand{\thesubsection}{\thesection.\arabic{subsection}}
\renewcommand{\theassumption}{A.\arabic{assumption}}
\renewcommand{\thetheorem}{A.\arabic{theorem}}
\renewcommand{\theproposition}{A.\arabic{proposition}}
\renewcommand{\thesection}{A}
\section{Appendix}
This appendix contains the proofs of all statements in the paper (Appendices \ref{subsec: proof}--\ref{subsec: CLT weighted}), some useful existing results (Appendix \ref{sec:useful}) and additional technical lemmas and derivations (Appendices \ref{sec:technical}--\ref{app: dev var}).

\subsection{Proofs of Statements}\label{subsec: proof}
\subsubsection{Proof of Lemma \ref{lem: unbiased example}}
\begin{proof}{}
    For any $\delta$, we have that 
    \[
    P\left(\psi(P)+n^{-1/2}\delta\in\psi_{n}\pm n^{-1/2}qA_n\right)=P\left(\left|\sqrt{n}\left(\psi_{n}-\psi(P)\right)\right|\in\delta\pm qA_n\right)\to P\left(\left|Z\right|\in\delta\pm qA\right).\]
From the independence between $Z$ and $A$, we have that for any $a\in R$, $P\left(\left|Z\right|\in\delta\pm qA|A=a\right) = P\left(\left|Z\right|\in\delta\pm qa\right)\leq P\left(\left|Z\right|\in 0\pm qa\right)$. Here, the last inequality follows from the fact that for the standard normal distribution function $\Phi$ and any $a\geq 0$, $\Phi(x+a)-\Phi(x-a)$ is maximized at $x=0$. Therefore, taking expectation, we get that $P\left(\left|Z\right|\in\delta\pm qA\right)\leq P\left(\left|Z\right|\in 0\pm qA\right).$ Therefore, from the above displayed limit, we get that 
$\lim_{n\to\infty}P(\psi(P)+n^{-1/2}\delta\in C)\leq1-\alpha$, which is the desired result. 
\hfill{$\Box$}\end{proof}

\subsubsection{Proof of Theorem \ref{thm: UMA to length}}

\begin{proof}{} 
    Observe that by Fubini (applicable since the integrand is
nonnegative)
\begin{align}
\sqrt{n}\mathbb{E}\left[U\left(\mathbf{Y}_{n}\right)-L\left(\mathbf{Y}_{n}\right)\right] & =\sqrt{n}\int1\{L\left(\mathbf{Y}_{n}\right)\leq x\leq U\left(\mathbf{Y}_{n}\right)\}dxdP\nonumber\\ 
 & =\sqrt{n}\int1\{L\left(\mathbf{Y}_{n}\right)\leq x\leq U\left(\mathbf{Y}_{n}\right)\}dPdx\nonumber\\ 
 & =\sqrt{n}\int P\left(L\left(\mathbf{Y}_{n}\right)\leq x\leq U\left(\mathbf{Y}_{n}\right)\}\right)dx \nonumber\\ 
(x=\psi(P)+n^{-1/2}\delta) & =\int P\left(L\left(\mathbf{Y}_{n}\right)\leq\psi(P)+n^{-1/2}\delta\leq U\left(\mathbf{Y}_{n}\right)\}\right)d\delta \label{eq: Fubuni}
\end{align}
Similarly, we have that 
\[
\sqrt{n}\mathbb{E}\left[U_1\left(\mathbf{Y}_{n}\right)-L_1\left(\mathbf{Y}_{n}\right)\right] = \int P\left(L_1\left(\mathbf{Y}_{n}\right)\leq\psi(P)+n^{-1/2}\delta\leq U_1\left(\mathbf{Y}_{n}\right)\}\right)d\delta. 
\]
By Fatou's lemma, from the above we get that 
\begin{align}
\liminf_{n\to\infty} \sqrt{n}\mathbb{E}\left[U_1\left(\mathbf{Y}_{n}\right)-L_1\left(\mathbf{Y}_{n}\right)\right] & \geq \int  \lim_{n\to\infty} P\left(L_1\left(\mathbf{Y}_{n}\right)\leq\psi(P)+n^{-1/2}\delta\leq U_1\left(\mathbf{Y}_{n}\right)\}\right)d\delta \\
& \geq \int\lim_{n\to\infty} P\left(L\left(\mathbf{Y}_{n}\right)\leq\psi(P)+n^{-1/2}\delta\leq U\left(\mathbf{Y}_{n}\right)\}\right)d\delta
\end{align}
Here, the last inequality follows from the definition of asymptotic UMA unbiasedness. Also note that we replaced $\liminf_{n\to\infty}$ with $\lim_{n\to\infty}$ on the RHS since the limits are assumed to exist. From \eqref{eq: Fubuni}, it suffices to show that 
\[
\lim_{n\to\infty}\int P\left(L\left(\mathbf{Y}_{n}\right)\leq\psi(P)+n^{-1/2}\delta\leq U\left(\mathbf{Y}_{n}\right)\}\right)d\delta = \int\lim_{n\to\infty} P\left(L\left(\mathbf{Y}_{n}\right)\leq\psi(P)+n^{-1/2}\delta\leq U\left(\mathbf{Y}_{n}\right)\}\right)d\delta
\]
We will show that 
\begin{equation}\label{eq: int convergence}
\lim_{n\to\infty}\int_{0}^{\infty}P\left(L\left(\mathbf{Y}_{n}\right)\leq\psi(P)+n^{-1/2}\delta\leq U\left(\mathbf{Y}_{n}\right)\}\right)d\delta = \int_{0}^{\infty}\lim_{n\to\infty} P\left(L\left(\mathbf{Y}_{n}\right)\leq\psi(P)+n^{-1/2}\delta\leq U\left(\mathbf{Y}_{n}\right)\}\right)d\delta
\end{equation}
and the integration from $-\infty$ to $0$ will follow from a similar argument. From Markov's inequality, since $\{\sqrt{n}(U(\mathbf{Y}_n)-\psi(P))_{+}\}_{n=1,2,\dots}$ has $1+\epsilon$ moment bounded by $M$, we have that $P\left(L\left(\mathbf{Y}_{n}\right)\leq\psi(P)+n^{-1/2}\delta\leq U\left(\mathbf{Y}_{n}\right)\}\right)\leq P\left(\delta\leq \sqrt{n}(U(\mathbf{Y}_n)-\psi(P))_{+}\}\right)\leq M\delta^{-(1+\epsilon)}$. Therefore, $P\left(L\left(\mathbf{Y}_{n}\right)\leq\psi(P)+n^{-1/2}\delta\leq U\left(\mathbf{Y}_{n}\right)\}\right)\leq g(\delta):=\begin{cases}
1 & \delta\leq1\\
M\delta^{-(1+\epsilon)} & \delta>1
\end{cases}$. Since $g(\cdot)$ is integrable, from dominated convergence, we have that \eqref{eq: int convergence} holds. 
\hfill{$\Box$}\end{proof}

\subsubsection{Proof of Theorem \ref{thm: std batching}}

\begin{proof}{}
Consider any asymptotically unbiased level $1-\alpha$ test $C$.
We have that
\begin{align*}
 & P(\psi(P)+n^{-1/2}\delta\in C(\mathbf{Y}_{n,B}))\\
= & P(L(\mathbf{Y}_{n,B})\leq\psi(P)+n^{-1/2}\delta\leq U(\mathbf{Y}_{n}))\\
(\text{denote \ensuremath{\mathbf{Z}_{n}=\sqrt{n}(\mathbf{Y}_{n}-\psi(P)1_{K})}})= & P(L(n^{-1/2}\mathbf{Z}_{n}+\psi(P)1_{K})\leq\psi(P)+n^{-1/2}\delta\leq U(n^{-1/2}\mathbf{Z}_{n}+\psi(P)1_{K}))\\
(\text{by the last condition of }L)= & P(L(n^{-1/2}\mathbf{Z}_{n})\leq n^{-1/2}\delta\leq U(n^{-1/2}\mathbf{Z}_{n}))\\
= & P\left(L(\mathbf{Z}_{n})\leq\delta\leq U(\mathbf{Z}_{n})\right)
\end{align*}
Therefore, by continuity of $L,U$ and the CLT for $\mathbf{Z}_{n}$
(technically, we also need that $P_{\mathbf{Z}\sim N(0,\sigma^{2}I)}(L\left(\mathbf{Z}\right)=\delta)=0$.
This is actually implied by our assumption that $L\left(x+c1_{K}\right)=L(x)+c$.
To see this, suppose on the contrary that there exists $\delta$ such
that $P_{\mathbf{Z}\sim N(0,\sigma^{2}I)}(L\left(\mathbf{Z}\right)=\delta)=c>0$,
then for any $\epsilon$, we have that $P_{\mathbf{Z}\sim N(\epsilon1_{K},\sigma^{2}I)}\left(L\left(\mathbf{Z}\right)=\delta+\epsilon\right)=P_{\mathbf{Z}\sim N(0,\sigma^{2}I)}\left(L\left(\mathbf{Z}\right)=\delta\right)=c$.
Note that for any $M>1$, we can find $\epsilon_{M}>0$ such that
when $\epsilon\in(0,\epsilon_{M})$, the likelihood ratio between
$N(\epsilon1_{K},\sigma^{2}I)$ and $N(0,\sigma^{2}I)$ can be bounded
within $\left(M^{-1},M\right)$. Therefore, $P_{\mathbf{Z}\sim N(0,\sigma^{2}I)}\left(L\left(\mathbf{Z}\right)=\delta+\epsilon\right)\geq P_{\mathbf{Z}\sim N(\epsilon1_{K},\sigma^{2}I)}\left(L\left(\mathbf{Z}\right)=\delta+\epsilon\right)/M=c/M$
for any $\epsilon\in(0,\epsilon_{M})$, which is a contradiction since
$(0,\text{\ensuremath{\epsilon_{M}}})$ contains infinite numbers)
\begin{align*}
\lim_{n\to\infty}P\left(L(\mathbf{Z}_{n})\leq\delta\leq U(\mathbf{Z}_{n})\right) & =P_{\mathbf{Z}\sim N(0,\sigma^{2}I)}(L(\mathbf{Z})\le\delta\leq U(\mathbf{Z}))\\
 & =P_{\mathbf{Z}\sim N(-\delta1_{K},\sigma^{2}I)}(L(\mathbf{Z})\le0\leq U(\mathbf{Z}))\\
 & =P_{\mathbf{Z}\sim N(-t\delta1_{K},t^{2}\sigma^{2}I)}(L(\mathbf{Z})\le0\leq U(\mathbf{Z})),\forall t\in\mathbb{R}
\end{align*}
where we used the homogenity of $L,U$ again in the last two equalities.

Therefore, the unbiasedness requires that for any $\delta\neq0$ and
$t\neq0$,
\[
P_{\mathbf{Z}\sim N(-t\delta1_{K},t^{2}\sigma^{2}I)}(L(\mathbf{Z})\le0\leq U(\mathbf{Z}))\leq1-\alpha
\]
and that when $\delta=0$, 
\[
P_{\mathbf{Z}\sim N(0_{K},t^{2}\sigma^{2}I)}(L(\mathbf{Z})\le0\leq U(\mathbf{Z}))\geq1-\alpha.
\]
Notice that for any $\mu,\sigma^{2}$, we can find $t,\delta$ such
that $(-t\delta1_{K},t^{2}\sigma^{2}I)=(\mu1_{K},\sigma^{2}I)$. Therefore,
from the above, we have that the test with rejection region $\{z:L(z)\leq0\leq U(z)\}^{c}$
is unbiased for testing $H_{0}:\mu=0$ in the family $\{N(\mu1_{K},\sigma^{2}I):\mu\in R,\sigma^{2}>0\}$.
From Section 5.2 of Lehmann, we know that a UMP unibased test is given
by rejection region $\{z:L_{B}(z)\leq0\leq U_{B}(z)\}^{c}$. Therefore,
from the definition of UMP unbiased test, we get that for any $\delta\neq0$,
\[
P_{Z\sim N(-\delta1_{K},\sigma^{2}I)}(L(\mathbf{Z})\le0\leq U(\mathbf{Z}))\geq P_{Z\sim N(-\delta1_{K},\sigma^{2}I)}(L_{B}(\mathbf{Z})\leq0\leq U_{B}(\mathbf{Z}))
\]
which implies
\begin{equation}
\lim_{n\to\infty}P\left(\psi(P)+n^{-1/2}\delta\in C(\mathbf{Y}_{n})\right)\geq\lim_{n\to\infty}P\left(\psi(P)+n^{-1/2}\delta\in C_{B}(\mathbf{Y}_{n})\right).\label{eq: cover ineq}
\end{equation}
This gives the desired result.

As seen in \eqref{eq: Fubuni}, we have that
\begin{align*}
\sqrt{n}\mathbb{E}\left[U\left(\mathbf{Y}_{n}\right)-L\left(\mathbf{Y}_{n}\right)\right] &  =\int P\left(L\left(\mathbf{Y}_{n}\right)\leq\psi(P)+n^{-1/2}\delta\leq U\left(\mathbf{Y}_{n}\right)\}\right)d\delta
\end{align*}
Then, by Fatou's lemma and (\ref{eq: cover ineq}), we have that
\begin{equation}
\liminf_{n\to\infty}\sqrt{n}\mathbb{E}\left[U\left(\mathbf{Y}_{n}\right)-L\left(\mathbf{Y}_{n}\right)\right]\geq\int P\left(L_{B}\left(\mathbf{Z}\right)\leq\psi(P)+n^{-1/2}\delta\leq U_{B}\left(\mathbf{Z}\right)\}\right)d\delta=\sqrt{n}\mathbb{E}_{\mathbf{Z}\sim N(0,\sigma^{2}I)}\left[U_{B}\left(\mathbf{Z}\right)-L_{B}\left(\mathbf{Z}\right)\right]\label{eq: lim inf length}
\end{equation}
On the other hand, for the length of the batching CI, we have that
\begin{align}
\sqrt{n}\mathbb{E}\left[U_{B}\left(\mathbf{Y}_{n}\right)-L_{B}\left(\mathbf{Y}_{n}\right)\right] & =\sqrt{n}t_{K-1;1-\alpha/2}\mathbb{E}S/\sqrt{K}\nonumber \\
 & =\frac{t_{K-1;1-\alpha/2}}{\sqrt{K(K-1)}}\mathbb{E}\sqrt{n\sum_{j=1}^{K}\left(\mathbf{Y}_{n,j}-\bar{\mathbf{Y}}_{n}\right)^{2}}\nonumber \\
 & =\frac{t_{K-1;1-\alpha/2}}{\sqrt{K(K-1)}}\mathbb{E}\sqrt{\sum_{j=1}^{K}\left(\mathbf{Z}_{n,j}-\bar{\mathbf{Z}}_{n}\right)^{2}}\label{eq: length finite}
\end{align}
Here, $\mathbf{Z}_{n,j},j=1,2,\dots,K$ denotes the $K$-th element
of $\mathbf{Z}_{n}$. Note that $\sqrt{\sum_{j=1}^{K}\left(\mathbf{Z}_{n,j}-\bar{\mathbf{Z}}_{n}\right)^{2}}\leq\sqrt{K}\sum_{j=1}^{K}\left|\mathbf{Z}_{n,j}-\bar{\mathbf{Z}}_{n}\right|\leq K\sqrt{K}\sum_{j=1}^{K}\left|\mathbf{Z}_{n,j}\right|$.
From the assumption, we have the uniform integrability of $\mathbf{Z}_{n,j}$
for any $j=1,2,\dots,K$. Therefore, we have the uniform integrability
of $\sqrt{\sum_{j=1}^{K}\left(\mathbf{Z}_{n,j}-\bar{\mathbf{Z}}_{n}\right)^{2}}$.
Therefore,
\[
\mathbb{E}\sqrt{\sum_{j=1}^{K}\left(\mathbf{Z}_{n,j}-\bar{\mathbf{Z}}_{n}\right)^{2}}\to \mathbb{E}_{\mathbf{Z}\sim N(0,\sigma^{2}I)}\sqrt{\sum_{j=1}^{K}\left(\mathbf{Z}_{j}-\bar{\mathbf{Z}}\right)^{2}}.
\]
Then from (\ref{eq: length finite}), we have that
\begin{align*}
 & \sqrt{n}\mathbb{E}\left[U_{B}\left(\mathbf{Y}_{n}\right)-L_{B}\left(\mathbf{Y}_{n}\right)\right]\\
\to & \frac{t_{K-1;1-\alpha/2}}{\sqrt{K(K-1)}}\mathbb{E}_{\mathbf{Z}\sim N(0,\sigma^{2}I)}\sqrt{\sum_{j=1}^{K}\left(\mathbf{Z}_{j}-\bar{\mathbf{Z}}\right)^{2}}\\
= & \sqrt{n}\mathbb{E}_{\mathbf{Z}\sim N(0,\sigma^{2}I)}\left[U_{B}\left(\mathbf{Z}\right)-L_{B}\left(\mathbf{Z}\right)\right].
\end{align*}
 Comparing with (\ref{eq: lim inf length}), we get the claim on the
length. 
\hfill{$\Box$}\end{proof}

\subsubsection{Proof of Lemma \ref{lem: corr formula}}

\begin{proof}{}
    Recall that $\hat{P}_{n}^j$ denotes the empirical distribution of the data in the
$j$-th batch. From the assumption, we have that 
\[
\sqrt{n}\left(\psi\left(\hat{P}_{n}^j\right)-\psi\left(P\right)-\mathbb{E}_{\hat{P}_{n}^j}G(X)\right)\to0.
\]
Therefore, 
\begin{equation}\label{eq: IF approx}
\sqrt{n}\left(\mathbf{Y}_{n,OB}-\psi\left(P\right)1_{K}-\mathbf{G}\right)\to0_{K}
\end{equation}
where $\mathbf{G}=\left(\mathbb{E}_{\hat{P}_{n}^1}G(X),\dots,\mathbb{E}_{\hat{P}_{n}^K}G(X)\right)$.
Since each batch has the form of $\cup_{i=1}^{m}\{X_{a_{i}n+1},X_{a_{i}n+2},\dots,X_{b_{i}n}\}$,
there exists $0\leq a_{1}^{\prime}\leq b_{1}^{\prime}\leq a_{2}^{\prime}\dots\leq a_{M}^{\prime}\leq b_{M}^{\prime}\leq1$
such that batch $j$ is $\cup_{i\in S_{j}}\{X_{a_{i}n+1},X_{a_{i}n+2},\dots,X_{b_{i}n}\}$
where $S_{j}$ is a subset of $\{1,2,\dots,M\}$. Let $\tilde{P}_{i}$
denote the empirical distribution of $\{X_{a_{i}n+1},X_{a_{i}n+2},\dots,X_{b_{i}n}\}$.
Then, from CLT, we have that $\left(\sqrt{n}\mathbb{E}_{\tilde{P}_{i}}G(X)\right)_{1\leq i\leq M}$
converges to a multivariate normal distribution with independent elements.
Noting that $\mathbf{G}$ is a linear transformation of $\left(\mathbb{E}_{\tilde{P}_{i}}G(X)\right)_{1\leq i\leq M}$,
we also have that $\sqrt{n}\mathbf{G}\Rightarrow N(0,\Sigma)$ for some
$\Sigma$. To find $\Sigma$, it suffices to study the covariances
for elements of $\mathbf{G}$. Suppose that $Var_{P}G(X)=\sigma^{2}$.
Then, we have that $Var\left(\sqrt{n}\mathbb{E}_{\hat{P}_{n}^j}G(X)\right)=\frac{n}{n\gamma_{j}^{2}}Var_P\left(G(X_{1})+\dots+G(X_{\gamma_{j}n}\right)=\frac{\sigma^2}{\gamma_{j}}$.
Moreover, $Cov\left(\sqrt{n}\mathbb{E}_{\hat{P}_{n}^j}G(X),\sqrt{n}\mathbb{E}_{\hat{P}_{n}^k}G(X)\right)=\frac{n}{n\gamma_{j}\gamma_{k}}\sum_{i=1}^{\beta_{ij}n}Var_{P}G(X_{i})=\frac{\sigma^{2}\beta_{ij}}{\gamma_{i}\gamma_{j}}$
(note that since $G(X_{i})$'s are independent, only the overlapping
part of the two batches will contribute to the covariance). Therefore,
$\sqrt{n}\mathbf{G}\Rightarrow N(0,\sigma^2 V)$, which implies the derised claim by Slutsky's theorem and \eqref{eq: IF approx}.
\hfill{$\Box$}\end{proof}

\subsubsection{Proof of Theorem \ref{thm: general batching}}

\begin{proof}{}
Replacing
$\sigma^{2}I$ with $\sigma^{2}\Sigma$ in the proof of Theorem \ref{thm: std batching},
we get that for any $\mu\neq0$ and $\sigma^{2}$, 
\[
P_{\mathbf{Z}\sim N(-\mu1_{K},\sigma^{2}\Sigma)}(L(\mathbf{Z})\le0\leq U(\mathbf{Z}))\leq1-\alpha
\]
and that 
\[
P_{\mathbf{Z}\sim N(0_{K},t^{2}\sigma^{2}\Sigma)}(L(\mathbf{Z})\le0\leq U(\mathbf{Z}))\geq1-\alpha.
\]
Therefore, the test with rejection region $\{z:L(z)\leq0\leq U(z)\}^{c}$
is unbiased for testing $H_{0}:\mu=0$ in the family $\{N(\mu1_{K},\sigma^{2}\Sigma):\mu\in R,\sigma^{2}>0\}$.
Then applying Theorem \ref{thm: UMP normal} (see Appendix \ref{subsec: UMP normal}), we get that the test
given by rejection region $\{z:L_{GS}(z)\leq0\leq U_{GS}(z)\}^{c}$
where $\left[L_{GS},U_{GS}\right]=CI_{GS}$ is UMP unbiased. Therefore,
again using the unbiasedness of $[L,U]$ as in Theorem \ref{thm: std batching},
we get that 
\begin{equation}
\lim_{n\to\infty}P\left(\psi(P)+n^{-1/2}\delta\in C(\mathbf{Y}_{n,GS})\right)\geq\lim_{n\to\infty}P\left(\psi(P)+n^{-1/2}\delta\in C_{GS}(\mathbf{Y}_{n,GS})\right).\label{eq: cover ineq general}
\end{equation}
Noting that $\sqrt{\frac{1}{\lambda}\left(\mathbf{Y}_{n,GS}-\frac{1_{K}^{\top}\Sigma^{-1}\mathbf{Y}_{n,GS}}{\lambda}1_{K}\right)^{\top}\Sigma^{-1}\left(\mathbf{Y}_{n,GS}-\frac{1_{K}^{\top}\Sigma^{-1}\mathbf{Y}_{n,GS}}{\lambda}1_{K}\right)}$
can also be bounded by a linear combination of $\left|\mathbf{Y}_{n,GS}\right|$,
we also have its uniform integrability, so the same proof as Theorem
\ref{thm: std batching} gives the claim on the interval length.
\hfill{$\Box$}\end{proof}

\subsubsection{Proof of Theorem \ref{thm: global}}

\begin{proof}{}
    First, we study the expected length of $CI_{GS}^{(\Sigma)}$. As shown in the proof of Theorem \ref{thm: general batching}, when $\sqrt{n}(\mathbf{Y}_n-\psi(P))\Rightarrow N(0,\sigma^2\Sigma)$ and the uniform integrability holds, we have the following asymptotic for the length
\begin{align}
 & \sqrt{n}\sqrt{\frac{1}{\lambda(K-1)}\left(\mathbf{Y}_{n,OB}-\frac{1_{K}^{\top}\Sigma^{-1}\mathbf{Y}_{n,OB}}{\lambda}1_{K}\right)^{\top}\Sigma^{-1}\left(\mathbf{Y}_{n,OB}-\frac{1_{K}^{\top}\Sigma^{-1}\mathbf{Y}_{n,OB}}{\lambda}1_{K}\right)}\notag\\
\Rightarrow & \sqrt{\frac{1}{\lambda(K-1)}\left(Z-\frac{1_{K}^{\top}\Sigma^{-1}Z}{\lambda}1_{K}\right)^{\top}\Sigma^{-1}\left(Z-\frac{1_{K}^{\top}\Sigma^{-1}Z}{\lambda}1_{K}\right)}\label{length comparison convergence}
\end{align}
where $Z\sim N(0,\sigma^2\Sigma)$. The proof of Theorem \ref{thm: UMP normal} implies that the limit \eqref{length comparison convergence} is equal in distribution to $\sigma\sqrt{\frac{\chi^2_{K-1}}{\lambda(K-1)}}$. Therefore, the asymptotic expected length of $CI_{GS}^{(\Sigma)}$ is proportional to $\sigma^2/\lambda$ as stage 1 varies. 

Now it suffices to study the condition for $\sigma^2/\lambda$ to achieve its minimum. Note that $1/\lambda$ is the optimal value of the optimization problem
\[
\min w^\top \Sigma w\ \ s.t.\ \ 1_K^\top w = 1.
\]
In other words, $\sigma^2/\lambda$ is the smallest possible variance among all possible affine combinations of $Z$. On the other hand, by the construction of this problem, we know that for any affine combination of $\mathbf{Y}_{n,OB}$, its asymptotic variance should be larger than or equal to the asymptotic variance of $\psi(\hat{P}_n)$ (which is $\sigma_0^2$). Therefore, $\sigma^2/\lambda\geq \sigma_0^2$, and the equality is achieved when there exists $w\in\mathbb{R}^K$ such that $1^\top w=1$ and $\sqrt{n}(w^\top\mathbf{Y}_n-\psi(P))\Rightarrow N(0,\sigma_0)$. Therefore, we get the dersired claim. 
\hfill{$\Box$}\end{proof}

\subsubsection{Proof of Corollary \ref{cor: equivalence} }

\begin{proof}{}
    To see the first part of the theorem, it suffices to check that all of $\mathbf{Y}_{n,B},\mathbf{Y}_{n,B}^{(\gamma)},\mathbf{Y}_{n,SJ},\mathbf{Y}_{n,CB},\mathbf{Y}_{n,OB}$ can asymptotically represent $\psi(\hat{P}_n)$ as an affine combination. For $\mathbf{Y}_{n,CB}$ and $\mathbf{Y}_{n,OB}$, this is true since they include $\psi(\hat{P}_n)$ as one of its coordinates. For $\mathbf{Y}_{n,B}$ and $\mathbf{Y}_{n,SJ}$, by symmetry, it is not hard to check that their average would asymptotically represent $\psi(\hat{P}_n)$. For $\mathbf{Y}_{n,B}^{(\gamma)}$, we can choose $w=(\gamma_1, \gamma_2, \dots, \gamma_K)$. This implies the first part of the theorem. 

    To see the asymptotic equivalence, noting that the point estimator in Theorem \ref{thm: general batching} is the min-variance estimator (since its variance is $\sigma^2/\lambda$, which equals the smallest possible variance we derived in the proof of Theorem \ref{thm: global}), we have that the point estimator of all these CIs differs from $\psi(\hat{P}_n)$ by $o(n^{-1/2})$. The claim on the distribution of asymptotic length follows directly from the proof of Theorem \ref{thm: global}.\hfill{$\Box$}
\end{proof}

\subsection{Estimating the Mean of Normal Distribution}\label{subsec: UMP normal}
\begin{theorem}
\label{thm: UMP normal}Given $Z=(Z_{1},\dots,Z_{K})^{\top}\sim N(\mu1_{K},\sigma^{2}\Sigma)$
where $\mu,\sigma^{2}$ belongs to a convex uncertainty set which
is not a subset of linear space of dimension $<2$, the UMP unbiased
hypothesis test for $H_{0}:$$\mu=\mu_{0}$ against $H_{1}:\mu\neq\mu_{0}$
is given by rejection region
\[
\left\{ z\in\mathbb{R}^{K}:\left|\frac{1_{K}^{\top}\Sigma^{-1}\left(z-\mu_{0}1_{K}\right)/\lambda}{\sqrt{\frac{1}{\lambda(K-1)}}\sqrt{\left(z-\frac{1_{K}^{\top}\Sigma^{-1}z}{\lambda}1_{K}\right)^{\top}\Sigma^{-1}\left(z-\frac{1_{K}^{\top}\Sigma^{-1}z}{\lambda}1_{K}\right)}}\right|>t_{K-1;1-\alpha/2}\right\} 
\]
where $\lambda=1_{K}^{\top}\Sigma^{-1}1_{K}$.
\end{theorem}

\begin{proof}{Proof}
The density can be written as
\begin{align*}
p(z,\mu,\sigma) & =C(\mu,\sigma)\exp\left\{ -\frac{1}{2\sigma^{2}}\left(z-\mu1_{K}\right)^{\top}\Sigma^{-1}\left(z-\mu1_{K}\right)\right\} \\
 & =C^{\prime}(\mu,\sigma)\exp\left\{ -\frac{z^{\top}\Sigma^{-1}z}{2\sigma^{2}}\right\} \exp\left\{ \frac{\mu1_{K}^{\top}\Sigma^{-1}z}{\sigma^{2}}\right\} 
\end{align*}
where $C(\sigma),C^{\prime}(\sigma)$ does not depend on $z$. Without
loss of generality, we study the case where $\mu_{0}=0$ and the general
case can be reduced to this case by defining $Z^{\prime}=Z-\mu_{0}1_{K}$.
Note that the test $\mu=0$ is equivalent to $\theta=0$ where $\theta:=\frac{\mu}{\sigma^{2}}$.
Define the nuisance parameter as $\vartheta:=\frac{1}{2\sigma^{2}}$.
Then, the density can be put as exponential family in the form of
Theorem \ref{thm: UMP theorem} where $U(Z)=1_{K}^{\top}\Sigma^{-1}Z$
and $T(Z)=Z^{T}\Sigma^{-1}Z$.

Let
\[
V(Z)=h(U(Z),T(Z)):=\frac{U(Z)}{\sqrt{T(Z)}}.
\]
Noting that $V(Z)$ does not change when $Z$ is replaced with $Z/\sigma$
whose distribution does not depend on $\vartheta$ (or $\sigma$)
when $\theta=0$, we have that 
\[
V(Z)=\frac{U(Z)}{\sqrt{T(Z)}}=\frac{1_{K}^{\top}\Sigma^{-1}Z/\sigma}{\sqrt{(Z/\sigma)^{T}\Sigma^{-1}(Z/\sigma)}}
\]
where the distribution of the RHS does not depend on $\vartheta$
when $\theta=0$. Therefore, applying Theorem \ref{thm: Basu corrollary}
we get the independence of $V(Z)$ and $T(Z)$. We can also check
that $h(u,t)=\frac{1}{\sqrt{t}}u$ is linear in $u$. Therefore, from
Theorem \ref{thm: UMP theorem}, we have that the UMP unbiased test
can be given as rejection region $\left\{ z:\left|V(z)\right|>v_{\alpha}\right\} $
where $v_{\alpha}$ is the $1-\alpha$ quantile of $V(Z)$ under $\theta=0$
(note that the condition $\mathbb{E}_{\theta_{0}}\left[V(\mathbf{X})\phi(V(\mathbf{X}))\right]=\alpha \mathbb{E}_{\theta_{0}}V(\mathbf{X})$
in Theorem \ref{thm: UMP theorem} can be checked by symmetry of the
distribution of $V(Z)$ about 0 when $\theta=0$. Indeed, both sides
are 0.).

For the next step, we write $\left\{ z:\left|V(z)\right|>v_{\alpha}\right\} $
in the form of a $t$-test. Define
\[
V^{\prime}(Z):=\frac{1_{K}^{\top}\Sigma^{-1}Z/\lambda}{\sqrt{\frac{1}{\lambda(K-1)}}\sqrt{\left(Z-\frac{1_{K}^{\top}\Sigma^{-1}Z}{\lambda}1_{K}\right)^{\top}\Sigma^{-1}\left(Z-\frac{1_{K}^{\top}\Sigma^{-1}Z}{\lambda}1_{K}\right)}}=\frac{U(Z)/\lambda}{\sqrt{\frac{1}{\lambda(K-1)}}\sqrt{T(Z)-U(Z)^{2}/\lambda}}.
\]
Noting that 
\[
V^{\prime}(Z)=\frac{U(Z)/\lambda\sqrt{T(Z)}}{\sqrt{\frac{1}{\lambda(K-1)}}\sqrt{1-U(Z)^{2}/\lambda T(Z)}}=\frac{V(Z)/\lambda}{\sqrt{\frac{1}{\lambda(K-1)}}\sqrt{1-V(Z)^{2}/\lambda}},
\]
we have that $\left|V^{\prime}(Z)\right|$ is an increasing function
of $\left|V(Z)\right|$. Therefore, the region $\left\{ z:\left|V(z)\right|>v_{\alpha}\right\} $
is the same as the region $\left\{ z:\left|V^{\prime}(z)\right|>v_{\alpha}^{\prime}\right\} $
where $v_{\alpha}^{\prime}$ is the $1-\alpha$ quantile of $V^{\prime}(Z)$.
Comparing with the desired claim, it remains to show that $v_{\alpha}^{\prime}=t_{K-1;1-\alpha/2}$,
which holds if we can show that $V^{\prime}(Z)\sim t_{K-1}$. To see
this, we first note that the numerator of $V^{\prime}$ has distribution
$N(0,\sigma^{2}/\lambda)$. Therefore, it suffices to show that $\sqrt{\lambda(K-1)}$
times the denominator of $V^{\prime}$ has distribution $\sqrt{\chi_{K-1}^{2}}$
and is independent of the numerator. Denote $\tilde{Z}=\Sigma^{-1/2}Z/\sigma$
which has standard $K$-dimensional normal distribution. We have the
sum-of-squares decomposition 
\begin{equation}
\tilde{Z}^{\top}\tilde{Z}=\left(\tilde{Z}-\frac{\Sigma^{-1/2}1_{K}1_{K}^{\top}\Sigma^{-1/2}\tilde{Z}}{\lambda}\right)^{\top}\left(\tilde{Z}-\frac{\Sigma^{-1/2}1_{K}1_{K}^{\top}\Sigma^{-1/2}\tilde{Z}}{\lambda}\right)+\left(1_{K}^{\top}\Sigma^{-1/2}\tilde{Z}\right)^{2}/\lambda\label{sum of squares 1}
\end{equation}
 Therefore, by Cochran's theorem and Lemma 1, we have that the first
term on the RHS above has distribution $\chi_{K-1}^{2}$. This is
equivalent to the claim that $\sqrt{\lambda(K-1)}$ times the denominator
of $V^{\prime}$ has distribution $\sqrt{\chi_{K-1}^{2}}$. To see
the independence of the numerator and denominator of $V^{\prime}(Z)$,
we can apply Theorem \ref{thm: Basu corrollary} again with the role
of $\theta$ and $\vartheta$ switched using the fact that the denominator
of $V^{\prime}(Z)$ does not change when $Z$ is replaced with $Z-c1_{K}$
for any $c$. The details are deferred to Lemma \ref{lem: independence}.
\hfill{$\Box$}\end{proof}

\subsection{Optimality of Length Among CIs That Can Be Biased}\label{subsec: biased CI}
Consider the class of CI with form $[\psi(\hat{P}_n)+t_{\beta,K-1}g(\mathbf{Y}_n),\psi(\hat{P}_n)+t_{1-\alpha+\beta,K-1}g(\mathbf{Y}_n)]$ where $g$ is continuous and satisfies $g(c\mathbf{y})=cg(\mathbf{y})$ and $g(\mathbf{y}+c1_{K-1}) = g(\mathbf{y})$ for any constant $c\in\mathbb{R}$. Also, we focus on CIs whose batches has the form of Lemma \ref{lem: corr formula} (note that since $\psi(\hat{P}_n)$ is already used in the point estimator, here we assume that the first coordinate of $\mathbf{Y}_n$ is given by $\psi(\hat{P}_n)$). Denote this class as $\mathcal{H}_0$. Then, the conclusion in Theorem \ref{thm: global} continues to hold when we replace the unbiased class $\mathcal{H}_{GS}$ with $\mathcal{H}_0$ which also contains biased CIs. 

\begin{theorem}\label{thm: optimal biased}
    Under the same assumptions as in Theorem \ref{thm: global}, $CI_{GS}^{(\Sigma)}(\mathbf{Y}_n)$ constructed in Theorem \ref{thm: general batching} is asymptotically globally shortest unbiased in $\mathcal{H}_{0}$ if and only if there exists $w\in\mathbb{R}^K$ such that $1^\top w=1$ and $\sqrt{n}(w^\top\mathbf{Y}_n-\psi(P))\Rightarrow N(0,\sigma_0)$.
\end{theorem}

\begin{proof}{Proof}
    Note that $t_{\beta,K-1}-t_{1-\alpha+\beta,K-1}$ is minimized when $\beta=\alpha/2$ (this is a corollary of Proposition 4.3 of \cite{glynn1990simulation}, which is also not hard to be verified directly).  Therefore, letting $\mathcal{H}_0^\prime$ be the subset of $\mathcal{H}_0$ where $\beta=\alpha/2$, it sufficies to show that $\mathcal{H}_0^\prime\subset \mathcal{H}_{GS}$. From the homogeneity condition on $g$ as introduced above, we know that all CIs in $\mathcal{H}_0^\prime$ belong to the homogeneous class. To see that CIs in $\mathcal{H}_0^\prime$ are asymptotically unbiased, we can use Lemma \ref{lem: unbiased example}. Indeed, we have that $\sqrt{n}g(\mathbf{Y}_n)=\sqrt{n}\left(g(\mathbf{Y}_n-\psi(\hat{P}_n))\right)=g\left(\sqrt{n}\left(\mathbf{Y}_n-\psi(\hat{P}_n)\right)\right)$ from the assumption on $g$. Moreover, from Lemma \ref{lem: corr formula}, we have that the limiting distribution of $\sqrt{n}\left(\psi(\hat{P}_n)-\psi(P)\right)$ and $\sqrt{n}\left(\mathbf{Y}_n-\psi(\hat{P}_n)\right)$ is jointly normal. Moreover, from the assumption on influence function, the limiting distribution is the same as the limiting distribution of $\sqrt{n}\mathbb{E}_{\hat{P}_n}G(X)$ and $\left(\sqrt{n}\left(\mathbb{E}_{\hat{P}_{n}^i}G(X)-\mathbb{E}_{\hat{P}_n
    }G(X)\right)\right)_{1\leq i\leq K}$ where $\hat{P}_{n}^i$ is the empirical distribution of the data in the $i$-th batch. By conditioning on $\hat{P}_n$ first, it is not hard to see that $Cov\left(\sqrt{n}\mathbb{E}_{\hat{P}_n}G(X),\sqrt{n}\left(\mathbb{E}_{\hat{P}_{n}^i}G(X)-\mathbb{E}_{\hat{P}_n
    }G(X)\right)\right)=0$. Therefore, the limiting joint distribution of $\sqrt{n}\mathbb{E}_{\hat{P}_n}G(X)$ and $\left(\sqrt{n}\left(\mathbb{E}_{\hat{P}_{n}^i}G(X)-\mathbb{E}_{\hat{P}_n
    }G(X)\right)\right)$ are independent. Therefore, $\psi(\hat{P}_n)$ and $\sqrt{n}g(\mathbf{Y}_n)=g\left(\sqrt{n}\left(\mathbf{Y}_n-\psi(\hat{P}_n)\right)\right)$ are also asymptotically independent and we can get the unbiasedness from Lemma \ref{lem: unbiased example}. This concludes the proof. \hfill{$\Box$}
\end{proof}

\subsection{CLT for Weighted Cheap Bootstrap}\label{subsec: CLT weighted}
From \cite{Barbe1995}, we have a CLT for the weighted bootstrap:
\begin{theorem}[Adapted from Theorem 3.1 of \cite{Barbe1995}]\label{thm: CLT weighted CB}
    Suppose that the following conditions holds:
    \begin{enumerate}
        \item $\max \left(W_{i,n}-n^{-1}\right)^2/\sum_{i=1}^n \left(W_{i,n}-n^{-1}\right)^2 = o_p(1)$ as $n\to\infty$
        \item $\mathcal{W}_n$ and $W_{i,n}^{(b)}$'s are independent of $X_1,\dots,X_n$.
        \item $\psi(\cdot)$ is Fréchet differentiable at $P$ in the sense that $\psi(P+t(Q-P)) - \psi(P) = \mathbb{E}_Q IF + R(Q,P)$ where $\left\vert R(Q,P)\right\vert \leq d_{\mathcal{H}}(Q,P)$. Here, $d_{\mathcal{H}}(Q,P) = \sup_{h\in\mathcal{H}}\left\vert\mathbb{E}_{P-Q}h\right\vert$ and $\mathcal{H}$ is a function class that satistifies: \begin{enumerate}
            \item There exists an envelope $H$ such that $\mathbb{E}_P H^2(X)\leq\infty$
            \item For a constant $\gamma>0$ and any $x$, $H(x)\geq \gamma$
            \item There exists non-negative constants $A$ and $D$ depending only on $\mathcal{H}$ such that for any $\epsilon>0$  and distribution $Q$ satisfying 
 $E_Q H(X)< \infty$, \[N_1(\epsilon,Q,\mathcal{H})\leq \max\{A\mathbb{E}_Q (H(X)/\epsilon)^D,1\}. \]
 Here, $N_1(\epsilon,Q,\mathcal{H})$ is the $\epsilon$-covering number of $\mathcal{H}$ using $L_1$ distance under $Q$.
        \end{enumerate}
    \end{enumerate}

Then, \begin{equation}\label{eq: weighted CB}\sqrt{n}\left(\psi(P^{W,1}_n)-\psi(\hat{P}_n)\right)\vert X_1,\dots,X_n \Rightarrow N(0,\sigma_0^2 \sigma_W^2)\end{equation} in probability.
\end{theorem}
Therefore, from Proposition 1 of \cite{lam2022cheap},  we have the following joint CLT
\begin{align*}
 & \left(\sqrt{n}(\psi(\hat{P}_n)-\psi(P)),\sqrt{n}(\psi(P_{n}^{W,1})-\psi(\hat{P}_n)),\dots,\sqrt{n}(\psi(P_{n}^{W,K-1})-\psi(\hat{P}_n))\right)^{\top}\\
\Rightarrow & N(0,\sigma^{2}diag(1,\sigma_{W}^{2},\sigma_{W}^{2},\dots,\sigma_{W}^{2})).
\end{align*}

\subsection{Useful Theorems}\label{sec:useful}

Following the notations in statistics, a hypothesis test can be defined
as a critical function $\phi:\mathcal{X}\to[0,1]$ where $\phi(x)$
is the probability of rejection when sample $x$ is observed. The
tests defined in terms of rejection region $\mathcal{R}\subset\mathcal{X}$
in our main body can be considered as a special type of hypothesis
test where $\phi(x)=1\{x\in\mathcal{R}\}$.
\begin{theorem}
\label{thm: UMP theorem}(Theorem 5.1.1. of Lehmann et al. (2005))
Consider testing $H_{0}:\theta=\theta_{0}$ against $H_{1}:\theta\neq\theta_{0}$
using data $\mathbf{X}$. Suppose that the distribution of $\mathbf{X}$
is parameterized by $\theta$ and nuisance parameter $\vartheta$
with density given by the exponential family
\begin{equation}
p(x,\theta,\vartheta)=C(\theta,\vartheta)\exp\{\theta U(x)+\sum_{i=1}^{k}\vartheta_{i}T_{i}(x)\},(\theta,\vartheta)\in\Omega\label{eq: exponential family}
\end{equation}
Assume that the parameter space $\Omega$ is convex and is not contained
in a linear space of dimension $<k+1$. Suppose that there exists
a continuous random variable $V(\mathbf{X})=h(U(\mathbf{X}),T(\mathbf{X}))$
that is independent of $T(\mathbf{X})$ when $\theta=\theta_{0}$.
Then, the following test is UMP unbiased provided that $h(u,t)=a(t)u+b(t)$
with $a(t)>0$:
\[
\phi(V(x))=\begin{cases}
1 & \text{if }V(x)>C_{1}\text{ or }V(x)<C_{2}\\
\gamma_{i} & \text{if }V(x)=C_{i},i=1,2\\
0 & \text{if }C_{1}<V(x)<C_{2}
\end{cases}
\]
where $C$'s and $\gamma$'s are determined by $\mathbb{E}_{\theta_{0}}\phi(V(\mathbf{X}))=\alpha$
and $\mathbb{E}_{\theta_{0}}\left[V(\mathbf{X})\phi(V(\mathbf{X}))\right]=\alpha \mathbb{E}_{\theta_{0}}V(\mathbf{X})$.
\end{theorem}

\begin{theorem}
\label{thm: Basu corrollary}(Corollary 5.1.1. of Lehmann et al. (2005))
Let $\mathcal{P}$ be the exponential family obtained from (\ref{eq: exponential family})
by letting $\theta$ have some fixed value. Then a statistic $V$
is independent of $T$ for all $\vartheta$ provided that the distribution
of $V$ does not depend on $\vartheta$.
\end{theorem}

\begin{theorem}
(Cochran's theorem)Suppose that $U_{1},\dots,U_{K}$ are i.i.d. standard
normals and denote $U=(U_{1},\dots,U_{K})$. Suppose that we have
the sum-of-squares decomposition $U^{\top}U=\sum_{i=1}^{m}U^{\top}B_{i}U$
where $B_{i}$ is a square matrix with rank $r_{i}$. Then $U^{\top}B_{i}U$
follows a chi-square distribution with $r_{i}$ degrees of freedom.
\end{theorem}

\subsection{Technical Lemmas}\label{sec:technical}
\begin{lemma}
\label{lem: rank}Let $V$ and $\lambda$ be defined as in Theorem
\ref{thm: UMP normal}. The rank of matrix $I-\frac{V^{-1/2}1_{K}1_{K}^{\top}V^{-1/2}}{\lambda}$
is $K-1$.
\end{lemma}

\begin{proof}{Proof}
From the property that $rank(A)+rank(B)\geq rank(A+B)$, we have that
$rank\left(I-\frac{V^{-1/2}1_{K}1_{K}^{\top}V^{-1/2}}{\lambda}\right)\geq rank(I)-rank(\frac{V^{-1/2}1_{K}1_{K}^{\top}V^{-1/2}}{\lambda})=K-1$.

On the other hand, since $\left(I-\frac{V^{-1/2}1_{K}1_{K}^{\top}V^{-1/2}}{\lambda}\right)1_{K}=0$,
we have that $rank\left(I-\frac{V^{-1/2}1_{K}1_{K}^{\top}V^{-1/2}}{\lambda}\right)\geq rank(I)\leq K-1$.
Combining this with the lower bound, we get the desired result.
\hfill{$\Box$}\end{proof}
\begin{lemma}
\label{lem: independence}Let $Z\sim N(\mu1_{K},\sigma^{2}\Sigma)$.
We have that $1_{K}^{\top}\Sigma^{-1}Z$ and $f(Z)$ are independent
for any function $f$ such that $f(z)=f(z-c1_{K}),\forall c\in\mathbb{R}$.
\end{lemma}

\begin{proof}{Proof}
The density of $Z$ can be write as
\[
p(z,\mu,\sigma)=C^{\prime}(\sigma)\exp\left\{ -\frac{z^{\top}\Sigma^{-1}z}{2\sigma^{2}}\right\} \exp\left\{ \frac{\mu1_{K}^{\top}\Sigma^{-1}z}{\sigma^{2}}\right\} 
\]
where $C^{\prime}(\sigma)$ does not depend on $z$. Let $\theta=-\frac{1}{2\sigma^{2}}$
and $\vartheta=\frac{\mu}{\sigma^{2}}$, then we can write the above
density in the form of \ref{eq: exponential family} with $U(Z)=Z^{\top}\Sigma^{-1}Z$
and $T(Z)=1_{K}^{\top}\Sigma^{-1}Z$. For any fixed $\theta$ (and
thus $\sigma^{2}$), the distribution of $f(Z)$ does not depend on
$\vartheta$ (or $\mu$) because $f(Z)=f(Z-\mu1_{K})\stackrel{d}{=}f(Z^{\prime})$
where $Z^{\prime}\sim N(0_{K},\sigma^{2}\Sigma)$. Therefore, applying
Theorem \ref{thm: Basu corrollary}, we get the independence.
\hfill{$\Box$}\end{proof}
\subsection{Derivation of \eqref{eq: var sj}}\label{app: dev var}
\begin{align*}
 & \frac{1}{\lambda}\left(\mathbf{Y}_{n,OB}-\frac{1_{K}^{\top}V^{-1}\mathbf{Y}_{n,OB}}{\lambda}1_{K}\right)^{\top}V^{-1}\left(\mathbf{Y}_{n,OB}-\frac{1_{K}^{\top}V^{-1}\mathbf{Y}_{n,OB}}{\lambda}1_{K}\right)\\
= & \frac{(K-1)^2}{K}\left(\mathbf{Y}_{n,OB}-\frac{1_{K}1_{K}^{\top}\mathbf{Y}_{n,OB}}{K}\right)^{\top}\left(I-\frac{K-2}{K-1}1_{K}1_{K}^{\top}\right)\left(\mathbf{Y}_{n,OB}-\frac{1_{K}1_{K}^{\top}\mathbf{Y}_{n,OB}}{K}\right)\\
= & \frac{(K-1)^2}{K}\mathbf{Y}_{n,OB}^{\top}\left(I-\frac{1}{K}1_{K}1_{K}^{\top}\right)\left(I-\frac{K-2}{K-1}1_{K}1_{K}^{\top}\right)\left(I-\frac{1}{K}1_{K}1_{K}^{\top}\right)\mathbf{Y}_{n,OB}\\
= & \frac{(K-1)^2}{K}\left(\left(I-\frac{1}{K}1_{K}1_{K}^{\top}\right)\mathbf{Y}_{n,OB}\right)^{\top}\left(I-\frac{1}{K}1_{K}1_{K}^{\top}\right)\mathbf{Y}_{n,OB}\\
= & \frac{1}{K}\left(\left(1_{K}1_{K}^{T}-(K-1)I-\frac{1}{K}1_{K}1_{K}^{\top}\right)\mathbf{Y}_{n,OB}\right)^{\top}\left(1_{K}1_{K}^{T}-(K-1)I-\frac{1}{K}1_{K}1_{K}^{\top}\right)\mathbf{Y}_{n,OB}
\end{align*}
Here, in the second equality, we used the fact that $\left(I-\frac{1}{K}1_{K}1_{K}^{\top}\right)1_K=0$.

\end{document}